\documentclass{article}

 \PassOptionsToPackage{numbers, compress}{natbib}

\usepackage[preprint]{neurips_2019}

\usepackage[utf8]{inputenc} 
\usepackage[T1]{fontenc}    
\usepackage{hyperref}       
\usepackage{url}            
\usepackage{booktabs}       
\usepackage{amsfonts}       
\usepackage{nicefrac}       
\usepackage{microtype}   
\usepackage{enumitem}
\usepackage{amsmath, parallel}
\usepackage{amsthm,amssymb,xspace}
\usepackage{subcaption}
\usepackage{graphicx, color}
\usepackage[ruled,vlined]{algorithm2e}
\usepackage{times}
\usepackage{xr}
\usepackage{titlesec}

\allowdisplaybreaks
\newtheorem{theorem}{Theorem}[section]
\newtheorem{lemma}[theorem]{Lemma}

\newtheorem{proposition}[theorem]{Proposition}
\theoremstyle{definition}
\newtheorem{definition}[theorem]{Definition}

\newtheorem{remark}[theorem]{Remark}

\begin{document}

\title{Change-point detection in dynamic networks via graphon estimation}
%

%
%
\author{
Zifeng Zhao\\
Department of Information Technology, \\
Analytics, and Operations\\
University of Notre Dame \\
\texttt{zzhao2@nd.edu } \\
\And
Li Chen \\
Department of Applied and\\
Computational Mathematics and Statistics \\
University of Notre Dame \\
\texttt{lchen6@nd.edu} \\
\And
Lizhen Lin \\
Department of Applied and\\
Computational Mathematics and Statistics \\
University of Notre Dame \\
\texttt{lizhen.lin@nd.edu} \\
}

\maketitle

\begin{abstract}
We propose a general approach for change-point detection in dynamic networks. The proposed method is model-free and covers a wide range of dynamic networks. The key idea behind our approach is to effectively utilize the network structure in designing  change-point detection algorithms. This is done via an initial step of graphon estimation, where we propose a modified neighborhood smoothing~(MNBS) algorithm
for estimating the link probability matrices of a dynamic network. Based on the initial graphon estimation, we then develop a screening and thresholding algorithm for multiple change-point detection in dynamic networks. The convergence rate and consistency for the change-point detection procedure are derived as well as those for MNBS. When the number of nodes is large~(e.g., exceeds the number of temporal points), our approach yields a faster convergence rate in detecting change-points comparing with an algorithm that simply employs averaged information of the dynamic network across time.
Numerical experiments demonstrate robust performance of the proposed algorithm for change-point detection under various types of dynamic networks, and superior performance over existing methods is observed. A real data example is provided to illustrate the effectiveness and practical impact of the procedure.

\end{abstract}

\section{Introduction}\label{sec-intro}

The last few decades have witnessed rapid advancement in models, computational algorithms and theories for inference of networks. This is largely motivated by the increasing prevalence of network data in  diverse fields of science, engineering and society, and the need to extract meaningful scientific information out of these network data.  In particular, the emerged field of statistical network analysis has spurred development of  many statistical models such as latent space model \citep{hoff:latent}, stochastic block model and their variants \citep{holland1983stochastic, newmanblock,ball2011efficient, rohe_yu}, and associated algorithms \citep{smyth:spectral, ng01spectral,LuxburgTutorial07,amini2014semidefinite} for various inference tasks including link prediction, community detection and so on.  However, the existing literature has been mostly focused on the analysis of one (and often large) network. While inference of single network remains to be an important research area due to its abundant applications in social network analysis, computational biology and other fields, there is emerging need to be able to analyze a collection of multiple network objects \citep{eric-paper, paperwitheric, NIPS2017_7282,spec-sparse}, with one notable example being temporal or dynamic networks. For example, it has become standard practice in many areas of neuroscience~(e.g., neuro-imaging) to use networks to represent various notions of connectivity among regions of interest~(ROI) in the brain observed in a temporal fashion when the subjects are engaged in some learning tasks. Analysis of such data demands development of new network models and tools, and leads to a growing literature on inference of dynamic networks, see e.g., \citep{Sewell2015, Pensky2018, sarkar05a}.

Our work focuses on change-point detection in dynamic networks, which is an important yet less studied aspect of learning dynamic networks.
The key insight of our proposed approach is to effectively utilize the network structure for efficient change-point detection in dynamic networks. This is done by first performing graphon estimation~(i.e. link probability matrix estimation) for the dynamic network, which then serves as basis of a screening and thresholding procedure for change-point detection. For graphon estimation in dynamic networks, we propose a novel modified neighborhood smoothing~(MNBS) algorithm, where a faster convergence rate is achieved via simultaneous utilization of the network structure and repeated observations of dynamic networks across time.

Most existing literature on change-point detection in dynamic networks~\citep[e.g.,][]{Peel2015, Marangoni-Simonsen2015, Wang2017, Corneli2017} rely on specific model assumptions and only provide computational algorithms without theoretical justifications. In contrast, our method is nonparametric/model-free and thus can be applied to a wide range of dynamic networks. Moreover, we thoroughly study the consistency and convergence properties of our change-point detection procedure and provide its theoretical guarantee under a formal statistical framework. Numerical experiments on both synthetic and real data are conducted to further demonstrate the robust and superior performance of the proposed method.

The paper is organized as follows. Section \ref{sec:related} discusses related work. Section \ref{sec:methods} proposes an efficient graphon~(link probability matrix) estimation method -- MNBS, and Section \ref{sec:change} introduces our change-point detection procedure. Numerical study on synthetic and real networks is carried out in Section \ref{sec:simu}. Our work concludes with a discussion. Technical proofs, additional numerical studies and suggestions of an additional graphon estimator can be found in the supplementary material.

\section{Related work}\label{sec:related}
The focus of this paper is on change-point detection for dynamic networks. Since our work employs an initial step of graphon estimation, we first review some literature on graphon estimation. We then discuss related work on change-point detection in dynamic networks. 

\cite{nbhd} proposes a novel estimator for estimating the link probability matrix $P$ of an undirected network by neighborhood smoothing (NBS). The essential idea consists of the following: 
Given an adjacent matrix $A$, the link probability $P_{ij}$ between node $i$ and $j$ is estimated by 
\begin{align*}
\hat{P}_{ij}=\frac{\sum_{i'\in \mathcal N_i} A_{i'j}}{|\mathcal N_i|},
\end{align*}
where $\mathcal N_i$ is a certain set of neighboring nodes of node $i$, which consists of the nodes that exhibit similar connection patterns as node $i$. With a well-designed neighborhood adaptive to the network structure, the smoothing achieves an accurate estimation for $P$. NBS in \cite{nbhd} estimates $P$ with a single adjacency matrix $A$. For a dynamic network, a sequence of adjacency matrices $A^{(1)},\ldots, A^{(T)}$ is available, which provides extra information of the network. By aggregating information from repeated observations across time, in Section \ref{sec:methods}, we propose a modified NBS by carefully shrinking the neighborhood size, which yields a better convergence rate in estimating the link probability matrix $P$ and thus an improved rate in change-point detection.

The literature for change-point detection in dynamic networks consists of different approaches. One stream of literature takes the approach of converting the sequence of networks into a time series of scalar/vector values via feature extraction of the networks or subsampling of the nodes, and then applying traditional change-point detection techniques for time series \citep[e.g.,][]{Priebe2005, McCulloh2011, Koutra2016, Wang2017}. Though being model-free, one drawback of this approach is the potential loss of information, thus power, in the conversion process. Another stream of literature typically assumes a specific generative model of the networks and develops model-based hypothesis testing for change-point detection \citep[e.g.,][]{Moreno2013,Peel2015}. The downside of this approach is its limited applicability to various dynamic networks due to the specific model assumption.
Moreover, almost all the existing literature focus on developing computational algorithms without providing formal theoretical guarantees, making it difficult to assess the properties of the algorithms. One exception is \cite{Chen2015}, where a graph-based nonparametric testing procedure is proposed for change-point detection in a general data sequence such as networks and is shown to attain a pre-specified level of type-I error. 
Originally designed for single change-point detection, \cite{Chen2015} can be extended to multiple change-point detection combined with binary segmentation~\citep{Vostrikova1981}. We compare our procedure with \cite{Chen2015} in the simulation study and real data analysis.

We also notice three recent works on change-point detection in dynamic networks. \cite{Bhattacharjee2018} considers change-point estimation in a dynamic stochastic block model~(SBM). \cite{Mukherjee2018} uses CUSUM statistics and discusses change-point detection for general networks. However, both works assume the prior knowledge of one and only one change-point. In contrast, in this paper, we consider the general setting where no prior information about the existence and number of change-points is given, which is more realistic and applicable to real data. \cite{Wang2018} also studies the problem of change point detection in dynamic networks and provides rigorous theoretical results. We note that \cite{Wang2018} is mainly a theoretical work in the sense that though a computational algorithm is provided in the paper, it does not offer any guidance on the (potentially delicate and challenging) choices of tuning parameters nor provide any simulation studies, making numerical comparisons infeasible. The philosophy of \cite{Wang2018} and our work are fundamentally different: our change-point algorithm is built upon a refined network estimation procedure and offers both change-point detection and post-hoc network estimation. To achieve an information-theoretic bound, \cite{Wang2018} essentially treats change-point detection for dynamic networks as change-point detection for high-dimensional vectors with independent Bernoulli entries. Their focus is solely on change-point detection and does not consider post-hoc network estimation.

Another related area of research is anomaly detection in dynamic networks, where the task is to detect short/abrupt deviation of the network behavior from its norm. This is not the focus of our paper and we refer the readers to \cite{Ranshous2015} for a comprehensive survey. 

\section{Modified neighborhood smoothing for dynamic networks  }\label{sec:methods}

In this section, we propose a neighborhood smoothing based estimator for link probability matrix estimation given repeated observations of an undirected network. 
This later serves as the basis for the proposed algorithm of multiple change-point detection for dynamic networks in Section \ref{sec:change}.

The basic setting is as follows. Given a network with a link probability matrix $P$, assume one observes independent (symmetric) adjacency matrices $A^{(t)}$ ($t=1,\ldots,T$) such that $A_{ij}^{(t)}\sim \text{Bernoulli}(P_{ij})$ for $i\leq j$, independently.
Based on the repeated observations $\{A^{(t)}\}_{t=1}^T$, we want to estimate the link probability matrix $P$. Note that when $T=1$, this reduces to the classical problem of link probability matrix estimation for a network (e.g., \cite{chatterjee2015}, \cite{gao2015} and \cite{nbhd}).
In particular, the neighborhood smoothing~(NBS) proposed by \cite{nbhd} is a computationally feasible algorithm which enjoys competitive error rate and is demonstrated to work well for real networks. Motivated by NBS, we propose a modified neighborhood smoothing~(MNBS) algorithm, which incorporates the repeated observations of the network across time and thus further improves the estimation accuracy of NBS.

Let $\bar{A}=\sum_{t=1}^{T}A^{(t)}/T,$ we define the distance measure between node $i$ and $i'$ as in \cite{nbhd} such that
$\tilde{d}^2(i,i')=\max_{k\neq i,i'}|\langle \bar{A}_{i\cdot}-\bar{A}_{i'\cdot}, \bar{A}_{k\cdot}\rangle|$, where $A_{i\cdot}$ denotes the $i$th row of $A$ and $\langle \cdot\,\cdot \rangle$ denotes the inner product of two vectors.
Based on the distance metric, define the neighborhood of node $i$ as
\begin{align}
\label{mnbs_nbd}
\mathcal{N}_i=\left\{i'\neq i: \tilde{d}(i,i') \leq q_i(q) \right\},
\end{align}
where $q_i(q)$ denotes the $q$th quantile of the distance set $\left\{\tilde{d}(i,i'): i'\neq i \right\}$. Given neighborhood $\mathcal{N}_i$ for each node $i$, we define the modified neighborhood smoothing~(MNBS) estimator as
\begin{align}
\label{mnbs_estimator}
\tilde{P}_{ij}=\frac{\sum_{i'\in\mathcal{N}_i} \bar{A}_{i'j}}{|\mathcal{N}_i|}.
\end{align}

Note that $q$ is a tuning parameter and affects the performance of MNBS via a bias-variance trade-off. In \cite{nbhd}, where $T=1$, the authors set $q=C(\log n/n)^{1/2}$ for some constant $C>0$. Thus, for each node $i$, the size of its neighborhood $|\mathcal{N}_i|$ is roughly $C(n\log n)^{1/2}$.

For MNBS, we set $q=C\log n/(n^{1/2}\omega)$, where $\omega = \min(n^{1/2},(T\log n)^{1/2})$. When $T=1$, MNBS reduces to NBS. For $T>1$, we have $\log n/(n^{1/2}\omega) <  (\log n/n)^{1/2}$ and thus MNBS estimates $P$ by smoothing over a smaller neighborhood. From a bias-variance trade-off point of view, the intuition behind this modification is that we can shrink the size of $\mathcal{N}_i$ to reduce the bias of $\tilde{P}_{ij}$ introduced by neighborhood smoothing while the increased variance of $\tilde{P}_{ij}$ due to the shrunken neighborhood can be compensated by the extra averaged information brought by $\{A^{(t)}\}_{t=1}^T$ across time.

We proceed with studying theoretical properties of MNBS. We assume the link probability matrix $P$ is generated by a graphon $f: [0,1]^2\times \mathbb{N} \to [0,1]$ such that $f(x,y)=f(y,x)$ and 
\begin{align*}
P_{ij}=f(\xi_i, \xi_j), \text{  for }i,j=1,\ldots,n, \text{ and } \xi_i \overset{i.i.d.}{\sim} \text{Uniform}[0,1].
\end{align*}

As in \cite{nbhd}, we study properties of MNBS for a piecewise Lipschitz graphon family, where the behavior of the graphon function $f(x,y)$ is regulated in the following sense.

\begin{definition}[Piecewise Lipschitz Graphon Family]
	\label{piecewise_lipschitz}
	For any $\delta, L>0$, let $\mathcal{F}_{\delta;L}$ be a family of piecewise Lipschitz graphon functions $f: [0,1]^2\times \mathbb{N}\to [0,1]$ such that $(i)$ there exists an integer $K\geq 1$ and a sequence $0=x_0<\cdots<x_K=1$ satisfying $\min_{0\leq s\leq K-1}(x_{s+1}-x_s)>\delta$, and $(ii)$ both $|f(u_1,v)-f(u_2,v)|\leq L|u_1-u_2|$ and $|f(u,v_1)-f(u,v_2)|\leq L|v_1-v_2|$ hold for all $u,u_1,u_2\in [x_s,x_{s+1})$, $v,v_1,v_2\in [x_t,x_{t+1})$ and $0\leq s,t\leq K-1.$
\end{definition}

As is illustrated by \cite{nbhd}, this graphon-based theoretical framework is a general model-free scheme that covers a wide range of exchangeable networks such as the commonly used Erd\"os-R\'enyi model and stochastic block model. See more detailed discussion about Definition \ref{piecewise_lipschitz} in \cite{nbhd}. 

For any $P, Q \in \mathbb{R}^{n\times n}$, define $d_{2,\infty}$, the normalized $2,\infty$ matrix norm, by
$$d_{2,\infty}(P,Q)=n^{-1/2}\|P-Q\|_{2,\infty}=\max_i n^{-1/2}\|P_{i\cdot}-Q_{i\cdot}\|_2.$$
We have the following error rate bound for MNBS. The sample size $T$ is implicitly taken as a function of $n$ and all limits are taken over $n\to \infty.$ 

\begin{theorem}[Consistency of MNBS]
	\label{MNBS_convergence}
	Assume $L$ is a global constant and $\delta=\delta(n,T)$ depends on $n,T$ satisfying $\lim_{n\to\infty} \delta/(n^{-1/2}\omega^{-1}\log n) \to \infty$ where $\omega =\omega(n,T) = \min(n^{1/2},(T\log n)^{1/2})$,  then the estimator $\tilde{P}$ defined in \eqref{mnbs_estimator}, with neighborhood $\mathcal{N}_i$ defined in \eqref{mnbs_nbd} and $q=B_0\log n/(n^{1/2}\omega)$ for any global constant $B_0>0$ satisfies
	\begin{align*}
	\max_{f\in \mathcal{F}_{\delta;L}} P\left( d_{2,\infty}(\tilde{P},P)^2 \geq C\frac{\log n}{n^{1/2}\omega} \right) \leq n^{-\gamma},
	\end{align*}
	for any $\gamma>0$, where $C$ is a positive global constant depending on $B_0$ and $\gamma.$
\end{theorem}

It is easy to see that the error rate in Theorem \ref{MNBS_convergence} also holds for the normalized Frobenius norm $d_F(\tilde{P},P)=n^{-1}\|\tilde{P}-P\|_F$. For $T=1$, $\omega=(\log n)^{1/2}$ and MNBS recovers the error rate $(\log n/n)^{1/2}$ of NBS in \cite{nbhd}. 
For $n>T$, which is the realistic case for repeated temporal observations of a large dynamic network, we can set $\omega=(T\log n)^{1/2}$ for MNBS and thus have
$$\max_{f\in \mathcal{F}_{\delta;L}}P\left(d_{2,\infty}(\tilde{P},P)^2\geq C\left(\frac{\log n}{nT} \right)^{1/2} \right)\leq n^{-\gamma}.$$
In other words, the network structure among $n$ nodes and the repeated observations along time dimension $T$ both help achieve better estimation accuracy for MNBS. In contrast, it is easy to see that a simply averaged $\bar{A}$ across time $T$ cannot achieve improved performance when $n$ increases.

\begin{remark}
	Another popular estimator of $P$, which is more scalable for large networks, is the USVT (Universal Singular Value Thresholding) proposed by \cite{chatterjee2015}. In Section \ref{sec:MUSVT} of the supplementary material, we propose a modified USVT for dynamic networks by carefully lowering the singular value thresholding level for $\bar A$. However, the convergence rate of the modified USVT is shown to be slower than MNBS. 
	Thus we do not pursue USVT-based change-point detection here.
\end{remark}

\section{MNBS-based multiple change-point detection}\label{sec:change}
In this section, we propose an efficient multiple change-point detection procedure for dynamic networks, which is built upon MNBS. We assume the observed dynamic network $\{A^{(t)}\}_{t=1}^T$ are generated by a sequence of probability matrices $\{P^{(t)}\}_{t=1}^T$ with $A_{ij}^{(t)}\sim \text{Bernoulli}(P_{ij}^{(t)})$ for $t=1,\ldots,T$. 
We are interested in testing the existence and further estimating the locations of potential change-points where $P^{(t)}\not = P^{(t+1)}$.
More specifically, we assume there exist $J$~($J\geq 0$) unknown change-points $\tau_0 \equiv 0<\tau_1<\tau_2<\ldots<\tau_{J}<T\equiv\tau_{J+1}$ such that
\begin{align*}
P^{(t)}=P_j, \text{ for } t=\tau_{j-1}+1,\ldots,\tau_{j}, \text{ and } j=1,\ldots,J+1.
\end{align*}
In other words, we assume there exist $J+1$ non-overlapping segments of $(1,\ldots,T)$ where the dynamic network follows the same link probability matrix on each segment and $P_j$ is the link probability matrix of the $j$th segment satisfying $P_j\not=P_{j+1}$. 
Denote $\mathcal{J}=\{\tau_1 <\tau_2 <\ldots<\tau_{J} \}$ as the set of true change-points and define $\mathcal{J}=\emptyset$ if $J=0$. Note that the number of change-points $J$ is allowed to grow with the sample size $(n,T)$.

\subsection{A  screening and thresholding  change point detection algorithm}
For efficient and scalable computation, we adapt a screening and thresholding algorithm that is commonly used in change-point detection for time series, see, e.g. \cite{Lee1996}, \cite{Niu2012}, \cite{Zou2014} and \cite{Yau2016}. The MNBS-based detection procedure works as follows.

\textbf{\textit{Screening}}: Set a screening window of size $h \ll T$. For each $t=h, \ldots, T-h$, we calculate a local window based statistic
\begin{align*}
D(t,h)=d_{2,\infty}(\tilde{P}_{t1,h},\tilde{P}_{t2,h})^2,
\end{align*}
where $\tilde{P}_{t1,h}$ and $\tilde{P}_{t2,h}$ are the estimated link probability matrices based on observed adjacency matrices $\{A^{(i)}\}_{i=t-h+1}^t$ and $\{A^{(i)}\}_{i=t+1}^{t+h}$ respectively by MNBS.

The local window size $h\ll T$ is a tuning parameter. In the following, we only consider the case where $(h \log n)^{1/2}\leq n^{1/2}$, which is the most likely scenario for real data applications and thus is more interesting. Therefore, for MNBS, we can set $\omega=\min(n^{1/2}, (h\log n)^{1/2})=(h\log n)^{1/2}$ and $q=B_0 (\log n)^{1/2}/(n^{1/2}h^{1/2})$. The result for $(h \log n)^{1/2}> n^{1/2}$ can be derived accordingly.

Intuitively, $D(t,h)$ measures the difference of the link probability matrices within a small neighborhood of size $h$ before and after $t$, where a large $D(t,h)$ signals a high chance of being a change-point. We call a time point $x$ an $h$-local maximizer of the function $D(t,h)$ if
\begin{align*}
D(x,h)\geq D(t,h), \text{ for all } t=x-h+1,\ldots, x+h-1.
\end{align*}

\textbf{\textit{Thresholding}}: Let $\mathcal{LM}$ denote the set of all $h$-local maximizers of the function $D(t,h)$, we estimate the change-points by applying a thresholding rule to $\mathcal{LM}$ such that
\begin{align*}
\hat{\mathcal{J}}=\{t| t \in \mathcal{LM} \text{ and } D(t,h)>\Delta_D\}, 
\end{align*}
where $\hat{\mathcal{J}}$ is the set of estimated change-points, $\hat{J}=\text{Card}(\hat{\mathcal{J}})$ and $\Delta_D$ is the threshold taking the form
\begin{align*}
\Delta_D=\Delta_D(h,n)= D_0 \frac{(\log n)^{1/2+\delta_0}}{n^{1/2}h^{1/2}},
\end{align*}
for some constants $D_0>0$ and $\delta_0>0$. Note that $\Delta_D$ dominates the asymptotic order of the MNBS estimation error $C(\log n)^{1/2}/(n^{1/2}h^{1/2})$  of $\tilde{P}_{t1,h}$ and $\tilde{P}_{t2,h}$ quantified by Theorem \ref{MNBS_convergence}.

The proposed algorithm is scalable and can readily handle change-point detection in large-scale dynamic networks as MNBS can be easily parallelized over $n$ nodes and the screening procedure is parallelizable over time $t=h,\ldots,T-h$.

\subsection{Theoretical guarantee of the change-point detection procedure}
We first define several key quantities that are used for studying theoretical properties of the MNBS-based change-point detection procedure. Define $\Delta^j=d_{2,\infty}(P_j,P_{j+1})^2$ for $j=1,\ldots,J$ and $\Delta^*=\min_{1\leq j\leq J}\Delta^j$, which is the minimum signal level in terms of $d_{2,\infty}$ norm. Also, define $D^*=\min_{1\leq j\leq J+1}(\tau_{j}-\tau_{j-1})$, which is the minimum segment length. We assume for each segment $j=1,\ldots,J+1$, its link probability matrix $P_j$ is generated by a piecewise Lipschitz graphon $f_j\in \mathcal{F}_{\delta;L}$ as in Definition \ref{piecewise_lipschitz}, where common constants $(\delta,L)$ are shared across segments. 

Note that $\mathcal{J}=\mathcal{J}(n,T)$, $J=J(n,T)$, $\Delta^*=\Delta^*(n,T)$, $D^*=D^*(n,T)$ and $\delta=\delta(n,T)$ are functions of $(n,T)$ implicitly. We have the following consistency result.

\begin{theorem}[Consistency of MNBS-based multiple change-point detection]
\label{cp_detection}
Assume there exists some $\gamma>0$ such that $T/n^\gamma \to 0$. Assume $L$ is a global constant and assume $\delta=\delta(n,T)$, $h=h(n,T)$ and $D^*=D^*(n,T)$ depend on $n,T$ satisfying $h< D^*/2$ and $\lim_{n\to\infty} \delta/(n^{-1}h^{-1}\log n)^{1/2} \to \infty$. 

If assume further that the minimum signal level $\Delta^*=\Delta^*(n,T)$ exceeds the detection threshold $\Delta_D=D_0 (\log n)^{1/2+\delta_0}/(n^{1/2}h^{1/2})$, i.e. $\lim_{n\to\infty}\Delta^*/\Delta_D>1$, then the MNBS-based change-point detection procedure with $q=B_0(\log n)^{1/2}/(n^{1/2}h^{1/2})$ satisfies 
\begin{align*}
\lim_{n\to\infty} P\left(\{\hat{J}=J\} \cap \{\mathcal{J}\subset: \hat{\mathcal{J}}\pm h \} \right)=1,
\end{align*}
for any constants $B_0, D_0, \delta_0>0,$ where $\mathcal{J}\subset: \hat{\mathcal{J}}\pm h$ means $\tau_j \in \{\hat{\tau}_j-h+1,\ldots,\hat{\tau}_j+h-1\} $ for $j=1,\ldots,J.$
\end{theorem}

In particular, Theorem \ref{cp_detection} gives a sure coverage property of $\hat{\mathcal{J}}$ in the sense that the true change-point set $\mathcal{J}$ is asymptotically covered by $\hat{\mathcal{J}}\pm h$ such that $\max_{j=1,\ldots,J}|\hat{\tau}_j-\tau_j|< h$. If $h/T\to 0$, Theorem \ref{cp_detection} implies $\hat{J}=J$ and $\max_{j=1,\ldots,J}|\hat{\tau}_j/T-\tau_j/T|<h/T\to 0$ in probability, which further implies consistency of the relative locations of the estimated change-points $\hat{\mathcal{J}}$.

A remarkable phenomenon occurs if the minimum signal level $\Delta^*$ is strong enough such that $\Delta^*>(\log n)^{1/2+\delta_0}/n^{1/2}$. Under such situation, we can set $h=1$ and by the sure coverage property in Theorem \ref{cp_detection}, the proposed algorithm recovers the exact location of the true change-points $\mathcal{J}$ \textbf{\textit{without}} any error. This is in sharp contrast to the classical result for change-point detection under time series settings~\citep{Yao1987}, where the optimal error rate for estimated change-point location is shown to be $O_p(1)$. This unique property is due to the fact that MNBS provides accurate estimation of the link probability matrix by utilizing the network structure within each network via smoothing.

\begin{remark}[Choice of tuning parameters]
	There are four tuning parameters of the MNBS-based detection algorithm: local window size $h$, neighborhood size $B_0$ and threshold size $(D_0, \delta_0)$. 
	For the window size $h$, a smaller $h$ gives a better convergence rate of change-point estimation and is more likely to satisfy the constraint that $h<D^*/2$. On the other hand, a smaller $h$ puts a higher requirement on the detectable signal level $\Delta^*$ since we require $\Delta^*/\Delta_D>1$ with a smaller $h$ implying a larger threshold $\Delta_D=D_0(\log n)^{1/2+\delta_0}/(n^{1/2}h^{1/2})$. The only essential requirement on $h$ is that $h<D^*/2$, i.e., the local window should not cover two true change-points at the same time. In practice, as long as a lower bound of $D^*$ is known, $h$ can be specified accordingly. For most applications, we recommend setting $h=\sqrt{T}.$
 	For the choice of $B_0$ and $(D_0,\delta_0)$, note that Theorem \ref{cp_detection} holds for any $B_0,D_0,\delta_0>0$. Thus the choice of $B_0$ and $(D_0,\delta_0)$ is more of a practical matter and specific recommendations are provided in Section \ref{sec:tuning_parameter} of the supplementary material, where MNBS is found to give robust performance across a wide range of tuning parameters.
\end{remark}

\begin{remark}[Separation measure of signals]
	To our best knowledge, all existing literature that study change-points of dynamic networks use Frobenius norm as the separation measure between two link probability matrices. We instead use $d_{2,\infty}$ norm, which in general gives much weaker condition than the one using Frobenius norm. See examples in Section \ref{sec:simu_synthetic}.
\end{remark}


\section{Numerical studies}\label{sec:simu}
In this section, we conduct numerical experiments to examine the performance of the MNBS-based change-point detection algorithm for dynamic networks. For comparison, the graph-based nonparametric testing procedure in \cite{Chen2015} is also implemented~(via R-package \textit{gSeg} provided by \cite{Chen2015}) with type-I error $\alpha=0.05$. We refer to the two detection algorithms as MNBS and CZ. 

To operationalize MNBS, we need to specify the neighborhood $q=B_0(\log n)^{1/2}/(n^{1/2}h^{1/2})$ and the threshold $\Delta_D= D_0(\log n)^{1/2+\delta_0}/(n^{1/2}h^{1/2})$. In total, there are four tuning parameters, $h$ for the local window size, $B_0$ for the neighborhood size, and $D_0$ and $\delta_0$ for the threshold size. In Section \ref{sec:tuning_parameter} of the supplementary material, we conduct extensive numerical experiments and provide detailed recommendations for calibration of the tuning parameters. In short, MNBS provides robust and stable performance across a wide range of tuning parameters. We refer readers to Section \ref{sec:tuning_parameter} of the supplementary material for detailed study of the tuning parameters. In the following, we recommend setting $h=\sqrt{T}$, $B_0=3$, $\delta_0=0.1$ and $D_0=0.25$.

\subsection{Performance on synthetic networks}\label{sec:simu_synthetic}
In this section, we compare the performance of MNBS and CZ under various synthetic dynamic networks that contain single or multiple change-points. We first define seven different stochastic block models~(SBM-I to SBM-VII), which we then use to build various dynamic networks that exhibit different types of change behavior.

Denote $K_B$ as the number of blocks in an SBM, denote $M(i)$ as the membership of the $i$th node and denote $\Lambda$ as the connection probability matrix between different blocks. Define $M_1(i)=\text{I}(1\leq i\leq \lfloor n/3 \rfloor)+2\text{I}(\lfloor n/3 \rfloor+1\leq i\leq 2\lfloor n/3 \rfloor) + 3\text{I}(2\lfloor n/3 \rfloor+1\leq i\leq n)$, where $\text{I}(\cdot)$ denotes the indicator function and $\lfloor x \rfloor$ denotes the integer part of $x$. Define
\vspace{-2mm}
\begin{table}[h]
	{\small
	\begin{tabular}{lll}
    $\Lambda_1=\left[\begin{matrix}
    0.6 & 0.6-\Delta_{nT} & 0.3\\
    0.6-\Delta_{nT} &0.6 &0.3\\
    0.3 &0.3&0.6
    \end{matrix}\right]$, &
     
    $\Lambda_2=\left[\begin{matrix}
    0.6 +\Delta_{nT} & 0.6 & 0.3\\
    0.6 &0.6+\Delta_{nT} &0.3\\
    0.3 &0.3&0.6
    \end{matrix}\right]$, &
    
    $\Lambda_3=\left[\begin{matrix}
    0.6  & 0.3\\
    0.3 &0.6
    \end{matrix}\right]$, \\
     
    $\Lambda_4=\left[\begin{matrix}
    0.6 +\Delta_{nT} & 0.6-\Delta_{nT} & 0.3\\
    0.6-\Delta_{nT} &0.6+\Delta_{nT} &0.3\\
    0.3 &0.3&0.6
    \end{matrix}\right]$,&
    
    $\Lambda_5=\left[\begin{matrix}
    0.6  & 0.6-\Delta_{nT}\\
    0.6-\Delta_{nT} & 0.6
    \end{matrix}\right]$. &
	\end{tabular} 
    }
\end{table}
\vspace{-3mm}

The seven SBMs are defined as:
\begin{align*}
&\textsc{[SBM-I]} && \hspace{-5mm} K_B=2, M(i)=\text{I}(1\leq i\leq 2\lfloor n/3 \rfloor)+2\text{I}(2\lfloor n/3 \rfloor+1\leq i\leq n), \Lambda=\Lambda_3.\\
&\text{[SBM-II]}  && \hspace{-5mm} K_B=2, M(i)=\text{I}(1\leq i\leq 2\lfloor n(1-\Delta_{nT})/3\rfloor)+2I(2\lfloor n(1-\Delta_{nT})/3\rfloor+1\leq i\leq n),\\ &\Lambda=\Lambda_3.\\
&\text{[SBM-III]} && \hspace{-3mm} K_B=3, M(i)=M_1(i), \Lambda=\Lambda_1(\Delta_{nT}).\\
&\text{[SBM-IV]} && \hspace{-3mm} K_B=3, M(i)=M_1(i), \Lambda=\Lambda_2(\Delta_{nT}).\\
&\text{[SBM-V]} && \hspace{-3mm} K_B=3, M(i)=M_1(i), \Lambda=\Lambda_4(\Delta_{nT}).\\
&\text{[SBM-VI]} && \hspace{-3mm} K_B=2, M(i)=\text{I}(1\leq i\leq 2\lfloor n/3 \rfloor)+2I(2\lfloor n/3 \rfloor+1\leq i\leq n), \Lambda=\Lambda_5(\Delta_{nT}).\\
&\text{[SBM-VII]} && \hspace{-3mm} K_B=2, M(i)=\text{I}(1\leq i\leq 2\lfloor n/3 \rfloor-1)+2I(2\lfloor n/3 \rfloor\leq i\leq n), \Lambda=\Lambda_5(\Delta_{nT}).
\end{align*}
\textbf{Dynamic networks with change-points:} Based on SBM-I to SBM-VII, we design five dynamic stochastic block models~(DSBM) with single or multiple change-points.

\vspace{-0.2mm}
\noindent[DSBM-I] (community merging) For $t=1,\ldots, T/2$, $P_1=$ SBM-III with $\Delta_{nT}=1/n^{1/6}/T^{1/8}$. For $t=T/2+1,\ldots,T$, $P_2=$ SBM-I.

\vspace{-0.2mm}
\noindent[DSBM-II] (connectivity changing) For $t=1,\ldots, T/2$, $P_1=$ SBM-III. For $t=T/2+1,\ldots,T$, $P_2=$ SBM-V. Set $\Delta_{nT}=1/n^{1/6}/T^{1/8}$.

\vspace{-0.2mm}
\noindent[DSBM-III] (community switching) For $t=1,\ldots, T/2$, $P_1=$ SBM-VI. For $t=T/2+1,\ldots,T$, $P_2=$ SBM-VII. Set $\Delta_{nT}=1/n^{1/6}/T^{1/8}$.

\vspace{-0.2mm}
\noindent[MDSBM-I] For $t=1,\ldots, T/4$, $P_1=$ SBM-II with $\Delta_{nT}=2/n^{1/3}/T^{1/4}$. For $t=T/4+1,\ldots,T/2$, $P_2=$ SBM-I. For $t=T/2+1,\ldots,3T/4$, $P_3=$ SBM-III with $\Delta_{nT}=1/n^{1/6}/T^{1/8}$. For $t=3T/4+1,\ldots,T$, $P_4=$ SBM-V with $\Delta_{nT}=1/n^{1/6}/T^{1/8}$.

\vspace{-0.2mm}
\noindent[MDSBM-II] For $t=1,\ldots, T/5$, $P_1=$ SBM-II with $\Delta_{nT}=2/n^{1/3}/T^{1/4}$. For $t=T/5+1,\ldots,2T/5$, $P_2=$ SBM-I. For $t=2T/5+1,\ldots,3T/5$, $P_3=$ SBM-III with $\Delta_{nT}=1/n^{1/6}/T^{1/8}$. For $t=3T/5+1,\ldots,4T/5$, $P_4=$ SBM-I. For $t=4T/5+1,\ldots,T$, $P_5=$ SBM-IV with $\Delta_{nT}=1/n^{1/6}/T^{1/8}$.

DSBM-I,II,III are three dynamic networks with different types of change at a single change-point $\tau_1=T/2$. MDSBM-I consists of 3 change-points $(\tau_1,\tau_2,\tau_3)=(T/4,T/2,3T/4)$ with the types of change being community switching, community splitting and connectivity changing respectively. MDSBM-II consists of 4 change-points $(\tau_1,\tau_2,\tau_3, \tau_4)=(T/5,2T/5,3T/5,4T/5)$ with types of change being community switching, community splitting, community merging and community splitting. Three additional DSBMs~(DSBM-IV to DSBM-VI) and their simulation results can be found in Section \ref{sec:supp_simu_DSBM} of the supplementary material.

The signal level $\Delta^*$ of each DSBM/MDSBM is controlled by $\Delta_{nT}$ and decreases as sample size $(T,n)$ grows. The detailed signal level is summarized in Tables \ref{tab: signal_level} and \ref{tab: signal_level_mcp} of the supplementary material. Note that the signal level measured by (normalized) Frobenius norm $d_F$ can be at a considerably smaller order than the one measured by (normalized) $d_{2,\infty}$ norm, indicating $d_{2,\infty}$ norm is more sensitive to changes. This phenomenon is especially significant for DSBM-III, where $d_{2,\infty}^2(P_1,P_2)=1/(n^{1/3}T^{1/4})$ and $d_{F}^2(P_1,P_2)=2/(n^{4/3}T^{1/4})$, since only one node switches membership after the change-point, making the change very challenging to detect in terms of Frobenius norm.

\textbf{Simulation result:} We vary $n=100, 500, 1000$ and $T=100, 500$. For each combination of sample size $(T,n)$ and DSBM/MDSBM, we conduct the simulation 100 times. Note that for multiple change-point scenarios, we conduct change-point analysis on MDSBM-I when $T=100$ and perform the analysis on MSDBM-II when $T=500$.

To assess the accuracy of change-point estimation, we use the Boysen distance as suggested in \cite{Boysen2009} and \cite{Zou2014}. Specifically, denote $\widehat{\mathcal{J}}_{nT}$ as the estimated change-point set and $\mathcal{J}_{nT}$ as the true change-point set, we calculate the distance between $\widehat{\mathcal{J}}_{nT}$ and $\mathcal{J}_{nT}$ via $\xi(\widehat{\mathcal{J}}_{nT}||\mathcal{J}_{nT})=\sup_{b\in \mathcal{J}_{nT}}\inf_{a\in\widehat{\mathcal{J}}_{nT}}|a-b| \text{ and }\xi(\mathcal{J}_{nT}||\widehat{\mathcal{J}}_{nT})=\sup_{b\in\widehat{\mathcal{J}}_{nT}}\inf_{a\in\mathcal{J}_{nT}}|a-b|$, which quantify the under-segmentation error and over-segmentation error of the estimated change-point set $\widehat{\mathcal{J}}_{nT}$, respectively. When $\mathcal{J}_{nT}\not= \emptyset$ and $\widehat{\mathcal{J}}_{nT}=\emptyset$, we define $\xi(\widehat{\mathcal{J}}_{nT}||\mathcal{J}_{nT})=\max_{\tau\in\mathcal{J}_{nT}}\tau$ and $\xi(\mathcal{J}_{nT}||\widehat{\mathcal{J}}_{nT})=\textbf{-}$.

The performance of MNBS and CZ are summarized in Table \ref{tab: onechangepoint}, where we report the average number of estimated change-points $\hat{J}$ and the average Boysen distance $\xi_1=\xi(\widehat{\mathcal{J}}_{nT}||\mathcal{J}_{nT})$ for under-segmentation error and $\xi_2=\xi(\mathcal{J}_{nT}||\widehat{\mathcal{J}}_{nT})$ for over-segmentation error over 100 runs.

The simulation results clearly indicate the superior performance of MNBS over CZ for all simulation scenarios in terms of both the number and accuracy of change-point estimation. For DSBM-I, CZ suffers from false positive detection as indicated by inflated over-segmentation error $\xi_2$ when $T=500$. For DSBM-II, CZ loses its power almost completely with $\hat{J}\approx 0$ when $T=100$ and produces false positive detection when $T=500$.

For DSBM-III, CZ suffers from low power especially for large $n$ due to the weak signal level of the change. For MDSBM-I, CZ underestimates the number of change-points due to the loss of power for detecting the connectivity changing at $\tau_3$, while for MDSBM-II, CZ suffers from false positive detection as indicated by inflated $\hat{J}$. In contrast, MNBS provides robust performance across all scenarios with more accurate estimated number of change-points $\hat{J}$ and smaller Boysen distances $\xi_1,\xi_2$ for both under and over-segmentation errors.

 
\vspace{-4mm}
\begin{table}[h]
	\centering
	\caption{Average number of estimated change-points $\hat{J}$ and Boysen distances $\xi_1$, $\xi_2$ by MNBS and CZ under single change-point and multiple change-point scenarios.}
	\label{tab: onechangepoint}
    {\small
	\begin{tabular}{lrrrrrrrrrrrr}
		\hline \hline
		MNBS & \multicolumn{3}{c}{DSBM-I} & \multicolumn{3}{c}{DSBM-II} & \multicolumn{3}{c}{DSBM-III} & \multicolumn{3}{c}{MDSBM-I/II} \\\hline
		$(T,n)$ &  $\hat{J}$ & $\xi_1$ & $\xi_2$ & $\hat{J}$ & $\xi_1$ & $\xi_2$ & $\hat{J}$ & $\xi_1$ & $\xi_2$ & $\hat{J}$ & $\xi_1$ & $\xi_2$  \\\hline
		$(100,100)$ & 1.00 & 0.32 & 0.32 & 1.00 & 0.16 & 0.16 & 1.01 & 0.86 & 1.25 & 3.00 & 0.47 & 0.47 \\ 
		$(100,500)$ & 1.00 & 0.09 & 0.09 & 1.00 & 0.03 & 0.03 & 1.00 & 0.65 & 0.65 & 3.00 & 0.10 & 0.10 \\ 
		$(100,1000)$ & 1.00 & 0.07 & 0.07 & 1.00 & 0.02 & 0.02  & 1.00 & 0.82 & 0.82   & 3.00 & 0.01 & 0.01\\
		$(500,100)$ & 1.00 & 1.46 & 1.46 & 1.00 & 1.26 & 1.26 & 1.02 & 2.47 & 3.53  & 4.00 & 2.91 & 2.91\\ 
		$(500,500)$ & 1.00 & 0.67 & 0.67 & 1.00 & 0.38 & 0.38 & 1.00 & 1.70 & 1.70 & 4.00 & 1.34 & 1.34 \\ 
		$(500,1000)$ & 1.00 & 0.45 & 0.45 & 1.00 & 0.19 & 0.19  & 1.00 & 2.11 & 2.11  & 4.00 & 0.42 & 0.42\\  \hline
		CZ & \multicolumn{3}{c}{DSBM-I} & \multicolumn{3}{c}{DSBM-II} & \multicolumn{3}{c}{DSBM-III} & \multicolumn{3}{c}{MDSBM-I/II} \\\hline
		$(T,n)$ &  $\hat{J}$ & $\xi_1$ & $\xi_2$ & $\hat{J}$ & $\xi_1$ & $\xi_2$ & $\hat{J}$ & $\xi_1$ & $\xi_2$ & $\hat{J}$ & $\xi_1$ & $\xi_2$ \\\hline
		$(100,100)$ & 1.12 & 0.00 & 3.38 & 0.03 & 49.49 & 30.50 & 0.61 & 26.66 & 9.26 & 2.31 & 26.64 & 4.29 \\ 
		$(100,500)$ & 1.11 & 0.00 & 2.66 & 0.00 & 50.00 & \textbf{-}& 0.24 & 40.63 & 8.55 & 2.19 & 26.56 & 3.31 \\ 
		$(100,1000)$ & 1.08 & 0.00 & 1.88 &  0.00 & 50.00 & \textbf{-}  & 0.10 & 45.65 & 6.50 & 2.26 & 26.76 & 3.94 \\ 
		$(500,100)$ & 1.15 & 0.00 & 22.75 & 2.36 & 9.60 & 112.24 & 1.09 & 20.11 & 24.59 & 6.53 & 9.51 & 45.66 \\ 
		$(500,500)$ & 1.21 & 0.00 & 21.59 & 2.46 & 9.91 & 117.94 & 0.64 & 122.70 & 40.88 & 8.14 & 10.94 & 41.78 \\
		$(500,1000)$ & 1.13 & 0.00 & 15.62 & 2.52 & 11.14 & 115.20 & 0.30 & 187.72 & 37.50 & 8.19 & 10.25 & 43.16 \\  \hline\hline	
	\end{tabular}
    }
\end{table}

\subsection{Real data analysis}
In this section, we apply MNBS and CZ to perform change-point detection for the MIT proximity network data. The data is collected through an experiment conducted by the MIT Media Laboratory during the 2004-2005 academic year \cite{Eagle2009}, where 90 MIT students and staff were monitored by means of their smart phone. The Bluetooth data gives a measure of the proximity between two subjects and can be used as to construct a link between them. Based on the recorded time of the Bluetooth scan, we construct a daily-frequency dynamic network among 90 subjects by grouping the links per day. The network extracted based on the Bluetooth scan is relatively dense~(see Figure 2 in the supplementary material).

There are in total 348 days from 07/19/2004 to 07/14/2005. The recommended $h$ is $h=\lfloor \sqrt{348} \rfloor=18$. For better interpretation, we set $h=14$, which corresponds to 2 weeks. CZ detects 18 change-points while MNBS gives 10 change-points. The detailed result is reported in Table \ref{tab: MIT_network}. The two algorithms give similar results for change-point locations. Notably CZ labels more change-points around the beginning and ending of the time period, which may be suspected as false positives. 
For robustness check, we rerun the analysis for MNBS and CZ with $h=7$, which corresponds to 1 week. CZ detects 28 change-points while MNBS detects 11 change-points, which further indicates the robustness of MNBS~(This result is reported in Section \ref{sec:supp_data} of the supplementary material). In Figure \ref{fig: MNBS_MIT}~(top) of the supplementary material, we plot the sequence of scan statistics $D(t,h)$ generated by MNBS, along with its $h$-local maximizers $\mathcal{LM}$ and estimated change-points $\widehat{\mathcal{J}}$. Figure \ref{fig: MNBS_MIT}(bottom) plots the time series of total links of the dynamic network for illustration purposes, where MNBS is seen to provide an approximately piecewise constant segmentation for the series.

\vspace{-4mm}
\begin{table}[h]
	\centering
	\caption{Estimated change-points by MNBS and CZ for MIT proximity network data.}
	\label{tab: MIT_network}
	{\small
	\begin{tabular}{lrrrrrrrrrrrrrrrrrr}
		\hline \hline
		CZ & 16 & 34 & 51 & 65 & 81 & 104 & 127 & 145 & 167 & 192 & 205 & 223 & 237 & 254 & 273 & 289 & 314 & 333 \\
		MNBS & - & - & 49 & 65 & 80 & 100 & - & 149 & - & 196 & - & - & 231 & 258 & 273 & 289 & - & - \\ \hline\hline
	\end{tabular}}
\end{table}

\section{Conclusion}\label{sec:conclusion}
We propose a model-free and scalable multiple change-point detection procedure for dynamic networks by effectively utilizing the network structure for inference. Moreover, the proposed approach is proven to be consistent and delivers robust performance across various synthetic and real data settings. One can leverage the insights gained from our work for other learning tasks such as performing hypothesis tests on populations of networks. One potential weakness of MNBS-based change-point detection is that it is currently not adaptive to the sparsity of the network, as graphon estimation by MNBS is not adaptive to sparsity. We expect to be able to build the sparsity parameter $\rho_n$ into our procedure by assuming the graphon function $f(x,y)=\rho_n f_0(x,y)$ with piecewise Lipschitz condition on $f_0$. This potentially allows to include $\rho_n$ into the error bound of MNBS (with Frobenius norm normalized by $\rho_n$) and subsequently to adjust detection threshold~(thus minimum detectable signal strength) depending on $\rho_n$.

%
\subsubsection*{Acknowledgments}
\vspace{-.5em}
Lizhen Lin would like to thank Soumendu Mukherjee for very useful discussions. She acknowledges the support from NSF grants IIS 1663870, DMS Career 1654579 and a DARPA grant N66001-17-1-4041.

\bibliographystyle{plain}
\bibliography{ref-LL-AA-networks}

\section{Supplementary material: additional numerical studies}\label{sec:numerical}
Section \ref{sec:SBMdefinition} defines additional SBMs/graphons used for numerical study of tuning parameter calibration and for specifying additional dynamic networks for change-point analysis. Section \ref{sec:tuning_parameter} provides detailed discussion and recommendation on the choices of tuning parameters for MNBS. Section \ref{sec:supp_simu_DSBM} gives additional numerical experiments for comparing the performance of MNBS and CZ on change-point detection. Section \ref{sec:supp_data} provides additional results on real data analysis.

\subsection{Additional SBMs and graphons}\label{sec:SBMdefinition}
In addition to the seven SBMs~(SBM-I to SBM-VII) defined in Section \ref{sec:simu_synthetic} of the main text, we further define four more SBMs~(SBM-VIII to SBM-XI) and three graphons~(Graphon-I to Graphon-III), which are later used for numerical study of tuning parameter calibration and for specifying additional dynamic networks for change-point analysis. Again, denote $K_B$ as the number of blocks/communities in an SBM, denote $M(i)$ as the membership of the $i$th node and denote $\Lambda$ as the connection probability matrix between different blocks. The three graphons are borrowed from \cite{nbhd}.


\noindent[SBM-VIII] $K_B=2$, $M(i)=\text{I}(1\leq i\leq \lfloor n^{3/4} \rfloor)+2\text{I}(\lfloor n^{3/4} \rfloor+1\leq i\leq n)$,\\
$\Lambda=\left[\begin{matrix}
0.6 & 0.3\\
0.3 &0.6
\end{matrix}\right]$.


\noindent[SBM-IX] $K_B=2$, $M(i)=\text{I}(1\leq i\leq \lfloor n^{3/4} \rfloor)+2\text{I}(\lfloor n^{3/4} \rfloor+1\leq i\leq n)$,\\
$\Lambda=\left[\begin{matrix}
0.6-\Delta_{nT} & 0.3\\
0.3 &0.6
\end{matrix}\right]$.




\noindent[SBM-X] $K_B=2$, $M(i)=\text{I}(1\leq i\leq \lfloor n/2 \rfloor)+2\text{I}(\lfloor n/2 \rfloor+1\leq i\leq n)$,\\
$\Lambda=\left[\begin{matrix}
0.6 & 0.6-\Delta_{nT}\\
0.6-\Delta_{nT} &0.6
\end{matrix}\right]$.

\noindent[SBM-XI] $K_B=2$, $M(i)=\text{I}(i \text{ is odd})+2\text{I}(i \text{ is even})$,\\
$\Lambda=\left[\begin{matrix}
0.6 & 0.6-\Delta_{nT}\\
0.6-\Delta_{nT} &0.6
\end{matrix}\right]$.



\noindent[Graphon-I] $f(u,v)=k/(K_B+1)$ if $(u,v)\in((k-1)/K_B,k/K_B),$ $f(u,v)=0.3/(K_B+1)$ otherwise; $K_B=\lfloor \log n\rfloor.$

\noindent[Graphon-II] $f(u,v)=\sin \{5\pi(u+v-1)+1 \}/2+0.5$.

\noindent[Graphon-III] $f(u,v)=(u^2+v^2)/3\cos\{1/(u^2+v^2)\} +0.15.$

\subsection{Calibration of tuning parameters}\label{sec:tuning_parameter}
In this section, we discuss and recommend the choices of tuning parameters for MNBS. To operationalize the MNBS-based detection procedure, we need to specify the neighborhood $q=B_0(\log n)^{1/2}/(n^{1/2}h^{1/2})$ and the threshold $\Delta_D=\Delta_D(h,n)= D_0(\log n)^{1/2+\delta_0}/(n^{1/2}h^{1/2})$. In total, there are four tuning parameters, $B_0$ for the neighborhood size, $h$ for the local window size, and $D_0$ and $\delta_0$ for the threshold size.

By Theorem \ref{cp_detection}, to achieve consistent detection of true change-points, the minimum signal level $\Delta^*$ needs to be larger than the threshold $\Delta_D$. Thus, in terms of minimum detectable signal, we prefer a larger $h$ and a smaller $\delta_0$, so that asymptotically we can achieve a larger detectable region. On the other hand, to achieve a tighter confidence region of estimated change-point locations, we prefer a smaller $h$, and we prefer a larger $\delta_0$ since it helps reduce false positives under small sample sizes. For finite sample, we recommend to set $\delta_0=0.1$ and $h=\sqrt{T}$, which makes the weakest detectable signal by MNBS to be of order $O\left(\frac{(\log n)^{0.6}}{n^{1/2}T^{1/4}}\right)$.

For the neighborhood size $B_0$, \cite{nbhd} demonstrates that the performance of the neighborhood-based estimation is robust to the choice of $B_0$ in the range of $[e^{-1}, e^{2}]$. Following \cite{nbhd}, we recommend to set $B_0=1, 2 \text{ or } 3$. Note that the number of neighbors in MNBS is $B_0(n\log n/h)^{1/2}$. To control the variance of MNBS, we suggest choosing a $B_0$ such that $B_0(n\log n/h)^{1/2}>10.$ For all the following simulations, we set $B_0=3$ where $B_0=1,2$ give similar numerical performance.

To study the sensitivity of $D_0$ w.r.t. false positives under finite sample, we simulate dynamic networks $\{A^{(t)}\}_{t=1}^T$ with \textbf{no} change-points from SBM-III,VIII,XI, and Graphons-I,II,III. We vary $D_0$ (thus the threshold $\Delta_D$) and examine the performance of MNBS. We vary the sample size at $n=100, 500, 1000$ and $T=100, 500$. 

For each combination of the sample size $(T,n)$ and the network model, we conduct the simulation 100 times. Figure \ref{fig: MNBS_nochange} reports the curves of average $|\hat{J}_{nT}-J_{nT}|$, which is the difference between estimated number change-points and true number of change-points, versus $D_0$ under the six different network models. Note that $J_{nT}\equiv 0$, thus $|\hat{J}_{nT}-J_{nT}|=\hat{J}_{nT}$ and the discrepancy represents the significance of false positive detection. As can be seen clearly, under various models and various sample sizes, the region with $D_0\geq 0.25$ controls the false positive reasonably well for MNBS. Thus, in practice we recommend setting $D_0=0.25$ for largest power. Note that out of the four stochastic block models~(SBM-III,VIII,XI and Graphon-I), the most challenging case is SBM-VIII, which may be due to its imbalanced block size.

\begin{figure}[h]
	\centering
	\begin{subfigure}[t]{0.32\textwidth}
		\centering
		\includegraphics[height=1.82in]{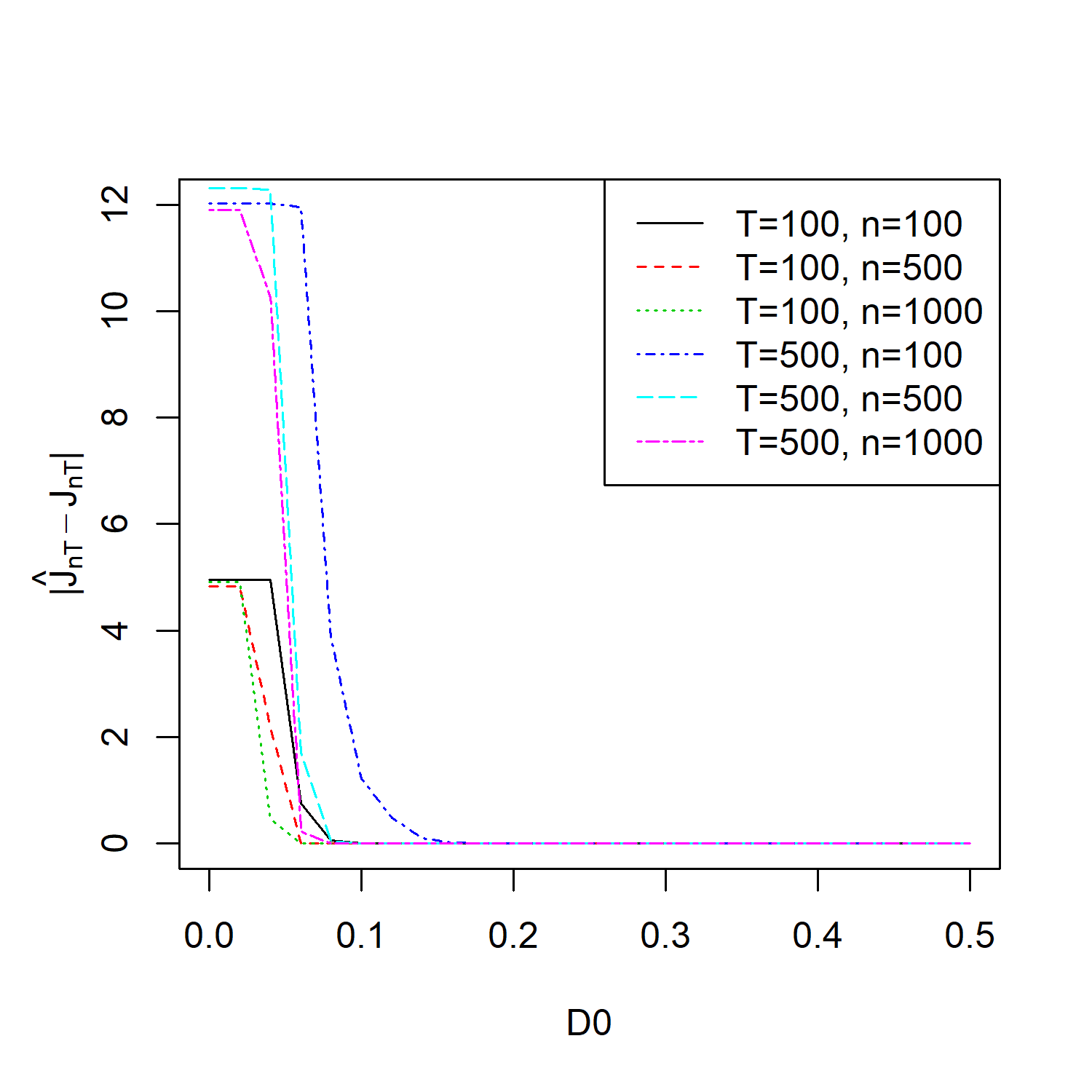}
		\caption{SBM-III}
	\end{subfigure}
	~ 
	\begin{subfigure}[t]{0.32\textwidth}
		\centering
		\includegraphics[height=1.82in]{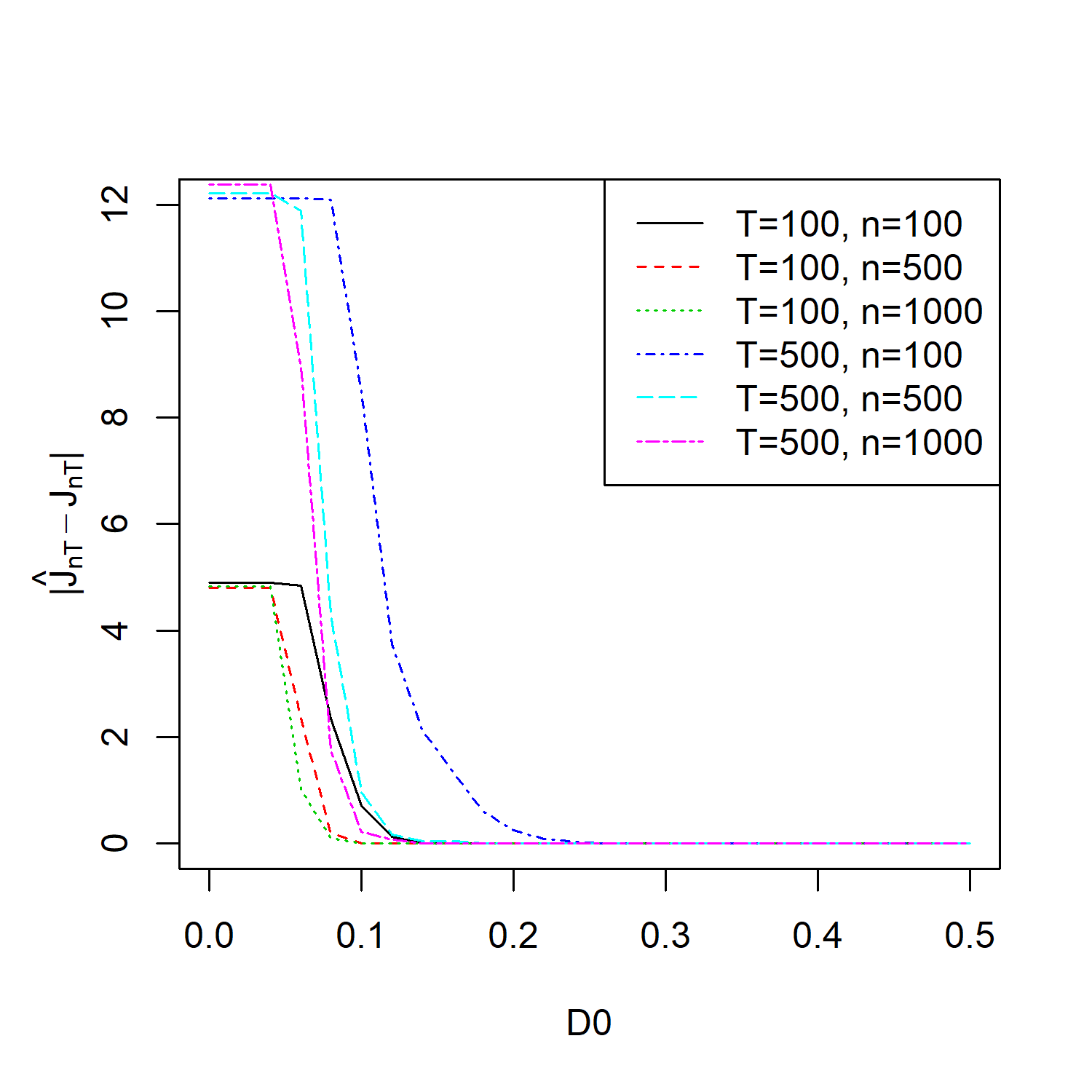}
		\caption{SBM-VIII}
	\end{subfigure}%
	~ 
	\begin{subfigure}[t]{0.32\textwidth}
		\centering
		\includegraphics[height=1.82in]{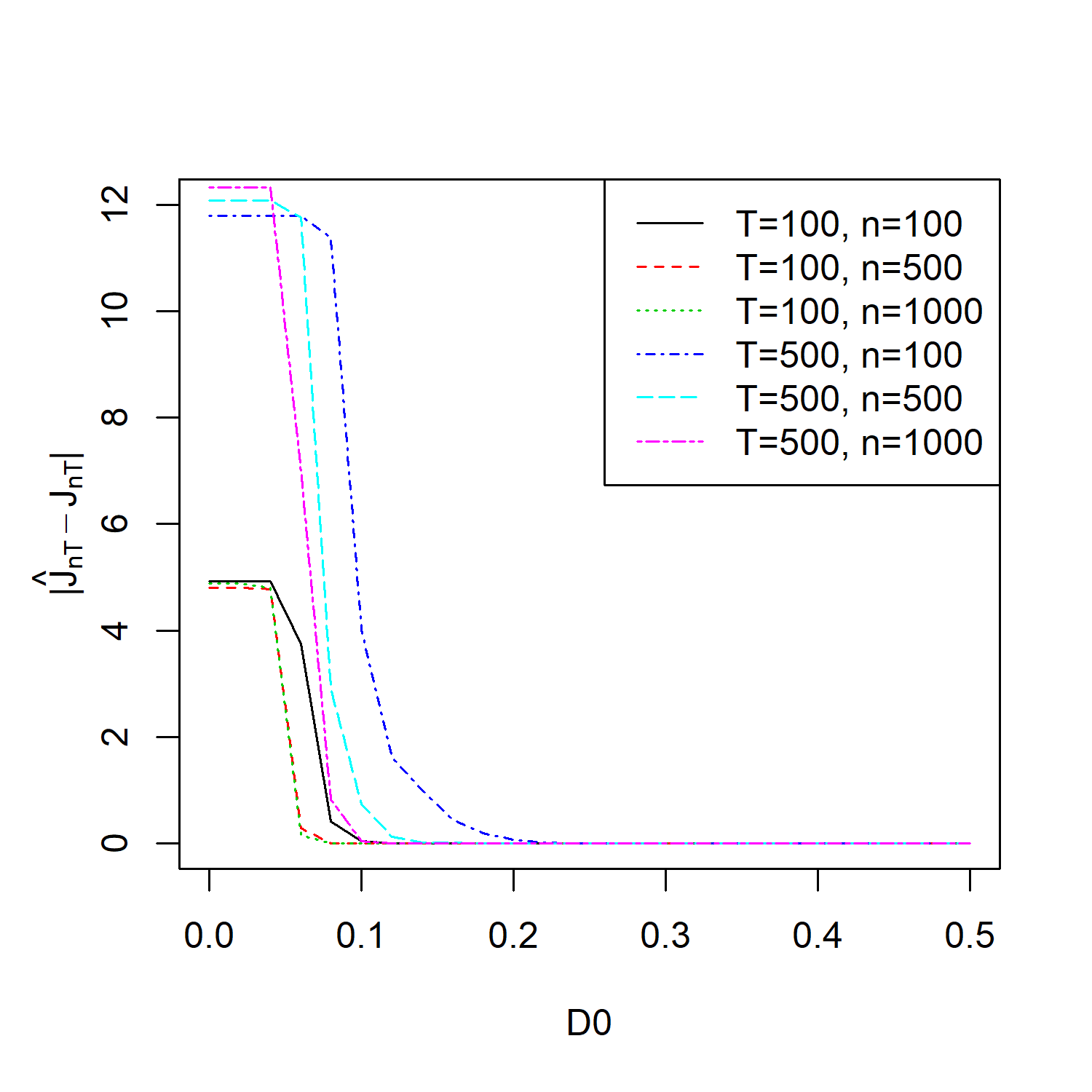}
		\caption{SBM-XI}
	\end{subfigure}
	
	\begin{subfigure}[t]{0.32\textwidth}
		\centering
		\includegraphics[height=1.82in]{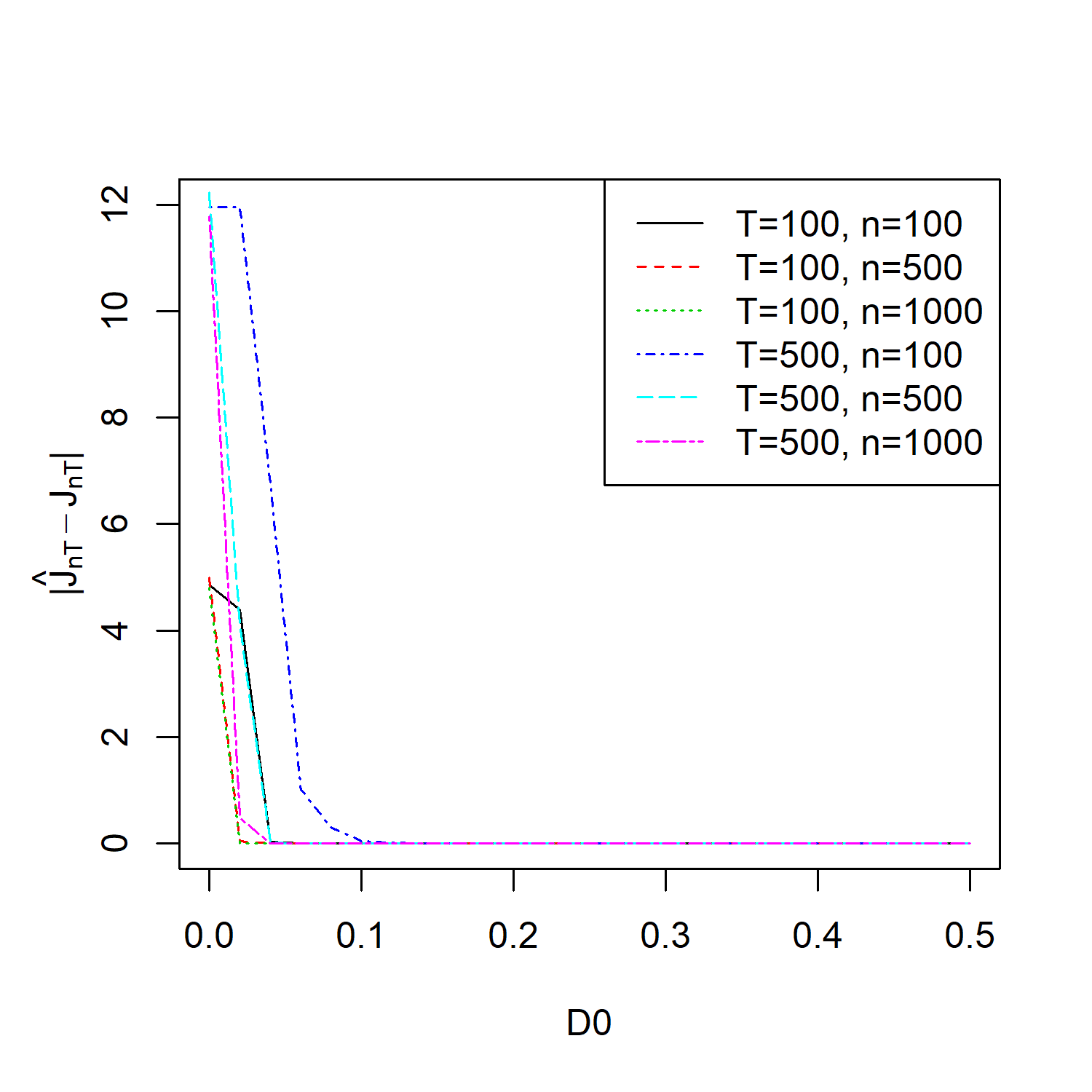}
		\caption{Graphon-I}
	\end{subfigure}
	~
	\begin{subfigure}[t]{0.32\textwidth}
		\centering
		\includegraphics[height=1.82in]{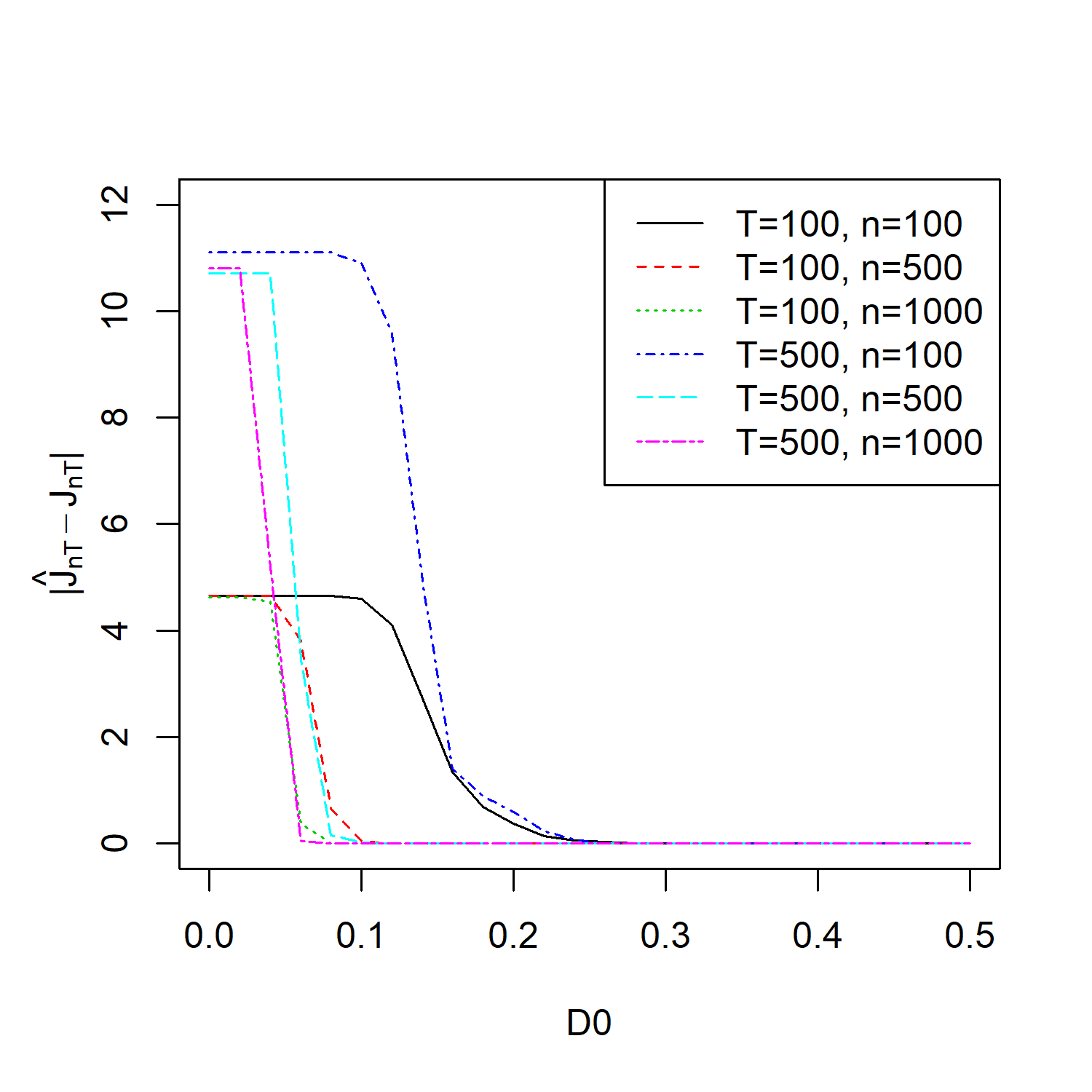}
		\caption{Graphon-II}
	\end{subfigure}
	~
	\begin{subfigure}[t]{0.32\textwidth}
		\centering
		\includegraphics[height=1.82in]{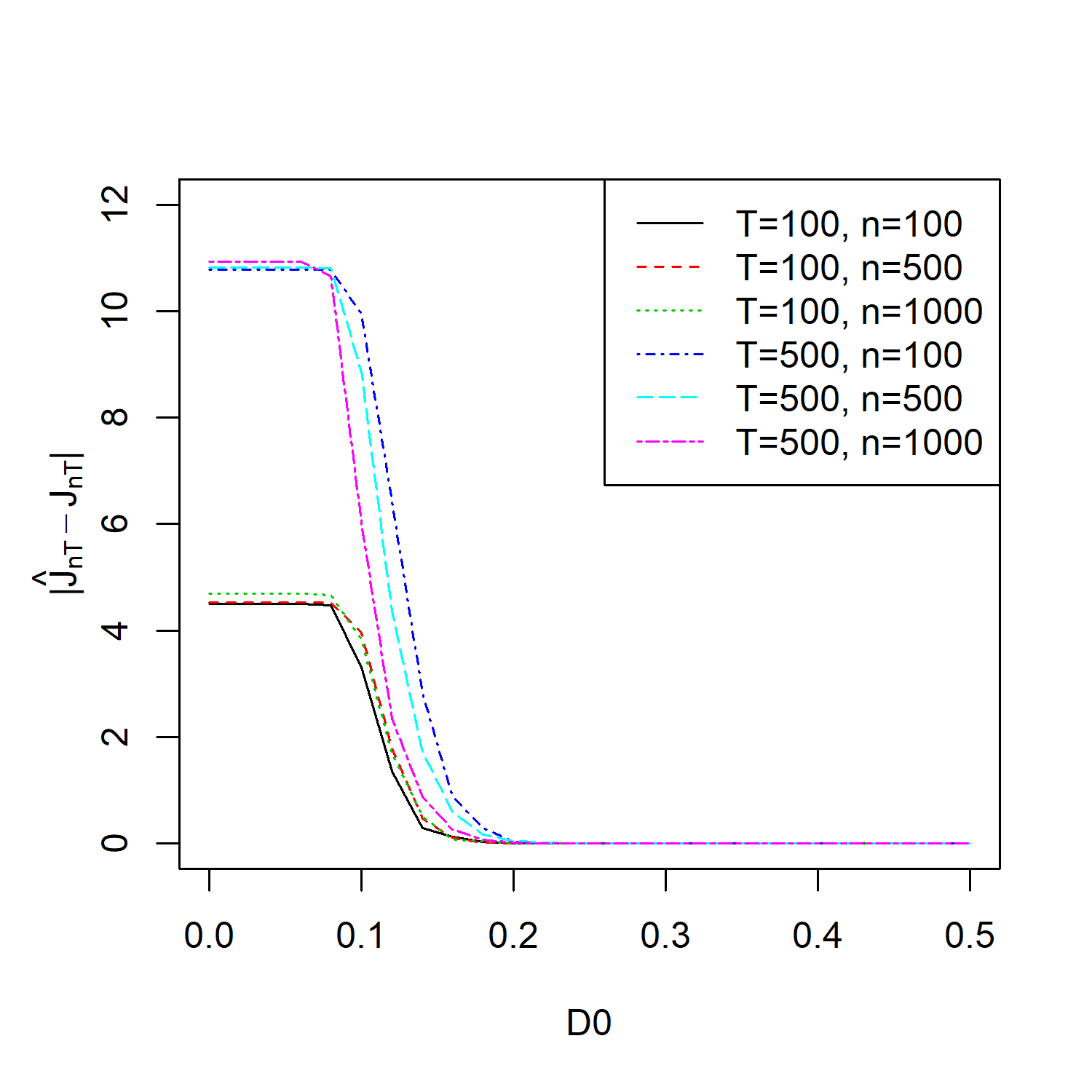}
		\caption{Graphon-III}
	\end{subfigure}
	\caption{Average number of estimated change-points $\hat{J}_{nT}$ by MNBS w.r.t. $D_0$ for six dynamic networks under no change-point scenarios.}
	\label{fig: MNBS_nochange}
\end{figure}

The average number of estimated change-points by MNBS with $D_0=0.25$ and by CZ with $\alpha=0.05$ are reported in Table \ref{tab: nochangepoint}. As can be seen, the performance of MNBS is reasonably well and improves as $n$ increases, with some false positives when $T=500, n=100$, i.e. the large $T$ and small $n$ case. The reason is due to the relatively large variance of MNBS incurred by the small size of neighborhood. As for CZ, the empirical type-I error is roughly controlled at the target level $0.05$ for most cases while experiencing some inflated levels under $T=500, n=100$.

\begin{table}[h]
	\centering
	\caption{Average number of estimated change-points by MNBS and CZ for six dynamic networks under no change-point scenarios.}
	\label{tab: nochangepoint}
	\begin{tabular}{lrrrrrr}
		\hline \hline
		MNBS & SBM-III & SBM-VIII & SBM-XI & Graphon-I & Graphon-II & Graphon-III\\\hline
		$T=100,n=100$ & 0.00 & 0.00 & 0.00 & 0.00 & 0.03 & 0.00 \\ 
		$T=100,n=500$ & 0.00 & 0.00 & 0.00 & 0.00 & 0.00 & 0.00 \\ 
		$T=100,n=1000$ & 0.00 & 0.00 & 0.00 & 0.00 & 0.00 & 0.00 \\ 
		$T=500,n=100$ & 0.00 & 0.02 & 0.00 & 0.00 & 0.03 & 0.00  \\ 
		$T=500,n=500$ & 0.00 & 0.00 & 0.00 & 0.00 & 0.00 & 0.00  \\ 
		$T=500,n=1000$ & 0.00 & 0.00 & 0.00 & 0.00 & 0.00 & 0.00 \\ \hline
		CZ  & SBM-III & SBM-VIII & SBM-XI & Graphon-I & Graphon-II & Graphon-III\\\hline
		$T=100,n=100$ & 0.05 & 0.01 & 0.05 & 0.05 & 0.03 & 0.04 \\ 
		$T=100,n=500$ & 0.01 & 0.08 & 0.02 & 0.05 & 0.05 & 0.03 \\ 
		$T=100,n=1000$ & 0.09 & 0.06 & 0.07 & 0.03 & 0.02 & 0.04 \\ 
		$T=500,n=100$ & 0.07 & 0.07 & 0.04 & 0.06 & 0.06 & 0.11 \\ 
		$T=500,n=500$ & 0.03 & 0.02 & 0.07 & 0.02 & 0.06 & 0.02 \\ 
		$T=500,n=1000$ & 0.07 & 0.03 & 0.02 & 0.03 & 0.05 & 0.04 \\ \hline \hline
	\end{tabular}
\end{table}

The numerical experiments demonstrate that MNBS provides robust and stable performance across a wide range of tuning parameters. To summarize, in practice, we recommend setting $h=\sqrt{T}$, $\delta_0=0.1$, $B_0\in \{1,2,3\}$ and $D_0\geq 0.25$. 

\subsection{Additional results on synthetic networks}\label{sec:supp_simu_DSBM}
In this section, we provide additional numerical experiments for comparing the performance of MNBS and CZ. Specifically, we design three additional dynamic stochastic block models~(DSBM-IV to DSBM-VI) and conduct change-point detection analysis.





\noindent[DSBM-IV] (community switching) For $t=1,\ldots, T/2$, $P_1=$ SBM-I. For $t=T/2+1,\ldots,T$, $P_2=$ SBM-II with $\Delta_{nT}=2/T^{1/4}/n^{1/3}$.

\noindent[DSBM-V] (connectivity changing) For $t=1,\ldots, T/2$, $P_1=$ SBM-VIII. For $t=T/2+1,\ldots,T$, $P_2=$ SBM-IX with $\Delta_{nT}=1/T^{1/8}$.

\noindent[DSBM-VI] (community switching) For $t=1,\ldots, T/2$, $P_1=$ SBM-X. For $t=T/2+1,\ldots,T$, $P_2=$ SBM-XI. Set $\Delta_{nT}=1/n^{1/6}/T^{1/8}$.

\textbf{Signal levels:} The signal levels $\Delta^*$ of DSBM-I to DSBM-VI are summarized in Table \ref{tab: signal_level}. The signal levels of MDSBM-I and MDSBM-II are summarized in Table \ref{tab: signal_level_mcp}. Again, we define the normalized $d_{2,\infty}$ norm as $d_{2,\infty}(P,Q)=n^{-1/2}\|P-Q\|_{2,\infty}=\max_i n^{-1/2}\|P_{i\cdot}-Q_{i\cdot}\|_2$ and the normalized Frobenius norm as $d_F(P,Q)=n^{-1}\|P-Q\|_F.$ Note that in general, the signal level measured by Frobenius norm $d_F$ is of considerably smaller order than the one measured by $d_{2,\infty}$ norm, indicating that $d_{2,\infty}$ norm is more sensitive to changes. This phenomenon is especially significant for DSBM-III, where only one node switches membership after the change-point, making the change very challenging to detect in terms of Frobenius norm.

\begin{table}[h]
	\centering
	\caption{Signal levels for six DSBMs under single change-point scenarios by $d_{2,\infty}^2$ and squared normalized Frobenius norm $d_{F}^2$ of $P_1-P_2$.}
	\label{tab: signal_level}
	\begin{tabular}{lllllll}
		\hline \hline
		$P_1-P_2$ & DSBM-I & DSBM-II & DSBM-III & DSBM-IV & DSBM-V & DSBM-VI	\\  \hline
		$d_{2,\infty}^2$ 	& $\dfrac{1}{3T^{1/4}n^{1/3}}$ & $\dfrac{1}{3T^{1/4}n^{1/3}}$ & $\dfrac{1}{n^{1/3}T^{1/4}}$ & $0.3^2$ & $\dfrac{1}{T^{1/4}n^{1/4}}$ & $\dfrac{1}{2T^{1/4}n^{1/3}}$ \\
		$d_{F}^2$ 	 	& $\dfrac{2}{9T^{1/4}n^{1/3}}$ & $\dfrac{2}{9T^{1/4}n^{1/3}}$ & $\dfrac{2}{n^{4/3}T^{1/4}}$ & $\dfrac{8\cdot0.3^2}{3T^{1/4}n^{1/3}}$ & $\dfrac{1}{T^{1/4}n^{1/2}}$ & $\dfrac{1}{2T^{1/4}n^{1/3}}$ \\
		\hline\hline
	\end{tabular}
\end{table}

\begin{table}[h]
	\centering
	\caption{Signal levels for two MDSBMs by $d_{2,\infty}^2$ and squared normalized Frobenius norm $d_{F}^2$ of $P_i-P_{i+1}, i=1,\ldots, J.$}
	\label{tab: signal_level_mcp}
	\begin{tabular}{lllll}
		\hline \hline
		& $\tau_1$ & $\tau_2$ & $\tau_3$ & $\tau_4$ 	\\  \hline
		MDSBM-I~($d_{2,\infty}^2$) 	&  $0.3^2$  & $\dfrac{1}{3T^{1/4}n^{1/3}}$ & $\dfrac{1}{3T^{1/4}n^{1/3}}$  & \textbf{-}\\
		MDSBM-I~($d_{F}^2$) 	& $\dfrac{8\cdot0.3^2}{3T^{1/4}n^{1/3}}$ & $\dfrac{2}{9T^{1/4}n^{1/3}}$ & $\dfrac{2}{9T^{1/4}n^{1/3}}$ & \textbf{-} \\
		MDSBM-II~($d_{2,\infty}^2$) 	&  $0.3^2$  & $\dfrac{1}{3T^{1/4}n^{1/3}}$ & $\dfrac{1}{3T^{1/4}n^{1/3}}$  &  $\dfrac{1}{3T^{1/4}n^{1/3}}$\\
		MDSBM-II~($d_{F}^2$) 	& $\dfrac{8\cdot0.3^2}{3T^{1/4}n^{1/3}}$ & $\dfrac{2}{9T^{1/4}n^{1/3}}$ & $\dfrac{2}{9T^{1/4}n^{1/3}}$ &  $\dfrac{2}{9T^{1/4}n^{1/3}}$\\
		\hline\hline
	\end{tabular}
\end{table}

\textbf{Simulation setting and result:} We set the sample size as $n=100, 500, 1000$ and $T=100, 500$. For each combination of sample size $(T,n)$ and DSBMs, we conduct the simulation 100 times. The performance of MNBS and CZ for DSBM-IV to DSBM-VI are reported in Table \ref{tab: onechangepoint}, where we report the average number of estimated change-points $\hat{J}$ and the average Boysen distance $\xi_1=\xi(\widehat{\mathcal{J}}_{nT}||\mathcal{J}_{nT})$ for under-segmentation error and $\xi_2=\xi(\mathcal{J}_{nT}||\widehat{\mathcal{J}}_{nT})$ for over-segmentation error. 
In general, both MNBS and CZ provide satisfactory performance for DSBM-IV to DSBM-VI, while MNBS offers superior performance with more accurate estimated number of change-points $\hat{J}$ and smaller Boysen distances $\xi_1,\xi_2$ for both under and over-segmentation errors.

\begin{table}[h]
	\centering
	\caption{Average number of estimated change-points $\hat{J}$ and Boysen distances $\xi_1$, $\xi_2$ by MNBS and CZ for DSBM-IV to DSBM-VI under single change-point scenarios.}
	\label{tab: onechangepoint}
	\begin{tabular}{lrrrrrrrrr}
		\hline \hline
		MNBS & \multicolumn{3}{c}{DSBM-IV} & \multicolumn{3}{c}{DSBM-V} & \multicolumn{3}{c}{DSBM-VI}  \\\hline
		&  $\hat{J}$ & $\xi_1$ & $\xi_2$ & $\hat{J}$ & $\xi_1$ & $\xi_2$ & $\hat{J}$ & $\xi_1$ & $\xi_2$ \\\hline
		$T=100,n=100$ & 1.00 & 0.12 & 0.12 & 1.00 & 0.00 & 0.00 & 1.00 & 0.02 & 0.02 \\ 
		$T=100,n=500$ & 1.00 & 0.01 & 0.01 & 1.00 & 0.00 & 0.00 & 1.00 & 0.00 & 0.00  \\ 
		$T=100,n=1000$ & 1.00 & 0.00 & 0.00 & 1.00 & 0.00 & 0.00 & 1.00 & 0.00 & 0.00  \\
		$T=500,n=100$ & 1.00 & 0.38 & 0.38 & 1.07 & 0.17 & 6.64 &  1.00 & 0.49 & 0.49 \\ 
		$T=500,n=500$ &  1.00 & 0.07 & 0.07 & 1.00 & 0.01 & 0.01 & 1.00 & 0.09 & 0.09  \\ 
		$T=500,n=1000$ & 1.00 & 0.02 & 0.02 & 1.00 & 0.00 & 0.00 & 1.00 & 0.09 & 0.09 \\  \hline
		CZ & \multicolumn{3}{c}{DSBM-IV} & \multicolumn{3}{c}{DSBM-V} & \multicolumn{3}{c}{DSBM-VI}  \\\hline
		&  $\hat{J}$ & $\xi_1$ & $\xi_2$ & $\hat{J}$ & $\xi_1$ & $\xi_2$ & $\hat{J}$ & $\xi_1$ & $\xi_2$ \\\hline
		$T=100,n=100$ & 1.14 & 0.00 & 3.46 & 1.18 & 0.00 & 3.45  & 1.13 & 0.00 & 2.71 \\ 
		$T=100,n=500$ & 1.10 & 0.00 & 2.39 & 1.05 & 0.00 & 1.18  & 1.14 & 0.00 & 2.68 \\ 
		$T=100,n=1000$ & 1.12 & 0.00 & 3.33 & 1.11 & 0.00 & 3.19 & 1.10 & 0.00 & 1.49 \\ 
		$T=500,n=100$ & 1.05 & 0.02 & 6.91 & 1.06 & 0.00 & 6.85  & 1.09 & 0.00 & 9.90 \\ 
		$T=500,n=500$ & 1.08 & 0.00 & 11.25 & 1.11 & 0.00 & 12.68 & 1.11 & 0.00 & 14.62 \\ 
		$T=500,n=1000$ & 1.13 & 0.00 & 10.24 & 1.06 & 0.00 & 7.97  & 1.10 & 0.00 & 10.70\\  \hline\hline
	\end{tabular}
\end{table}

\subsection{Additional results on real data analysis}\label{sec:supp_data}
For robustness check, we run the analysis for MNBS and CZ with $h=7$, which corresponds to 1 week. With $h=7$, CZ detects 28 change-points at $t=10, 16, 27, 34, 65, 73, 81, 94, 104, 116, 127, 139,$ $145, 167, 175, 181, 192, 205, 213, 223, 230, 237, 254, 273, 283, 289, 314, 333$ and MNBS detects 11 change-points at $t=50, 58, 66, 95, 135, 149, 196, 206, 223, 248, 256.$

In Figure \ref{fig: MNBS_MIT}~(top), we plot the sequence of scan statistics $D(t,h)$~(solid curve) generated by MNBS, along with the threshold $\Delta_D$~(horizontal line), the local-maximizers $\mathcal{LM}$~(red points) and the estimated change-points $\widehat{\mathcal{J}}$~(vertical dashed line). Figure \ref{fig: MNBS_MIT}(bottom) plots the time series of total links of the dynamic network for illustration purposes. As can be seen, MNBS provides an approximately piecewise constant segmentation for the series.

\begin{figure}[h]
	\centering
	\includegraphics[scale=0.5]{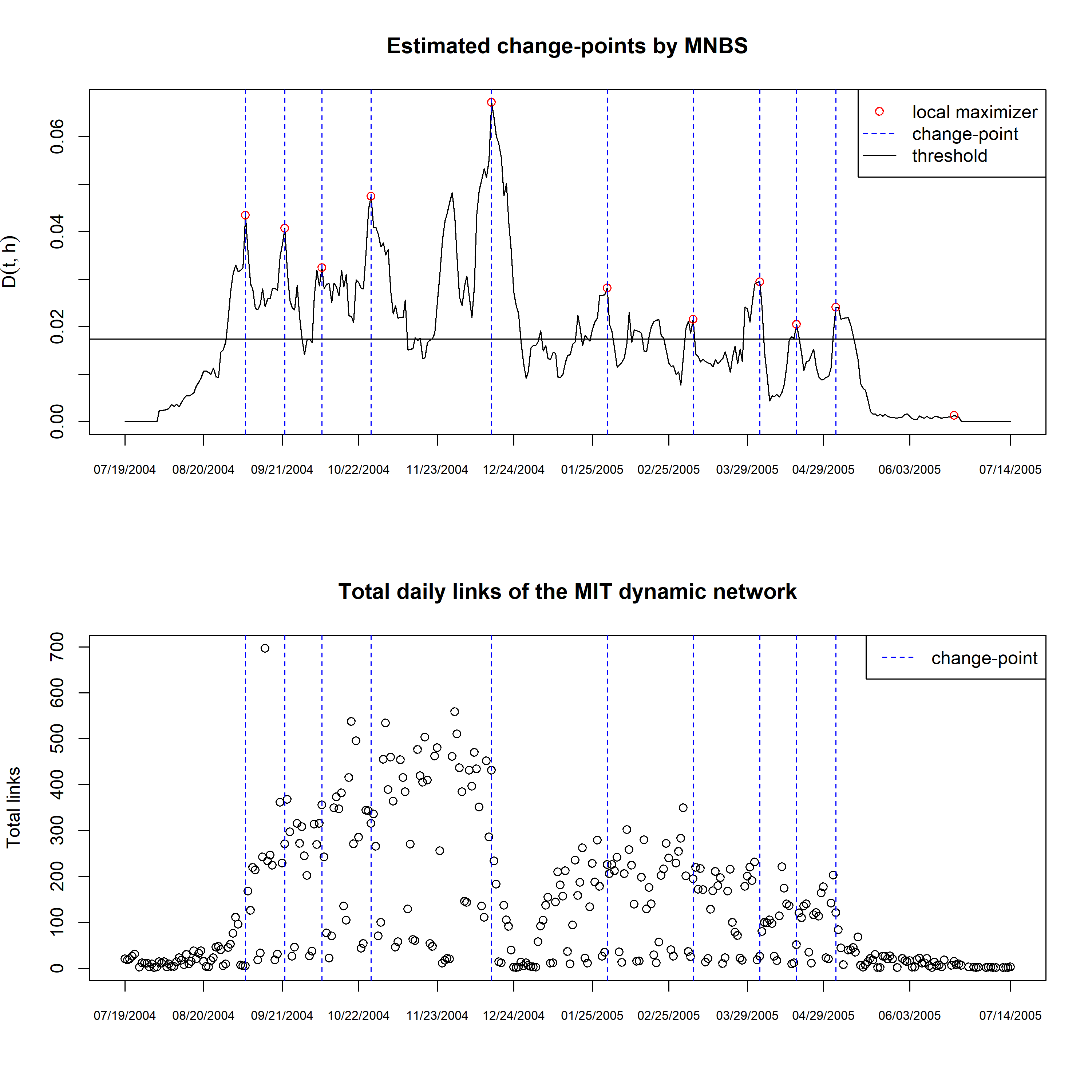}
	\caption{Scan statistics, local-maximizers and estimated change-points by MNBS for MIT network data.}
	\label{fig: MNBS_MIT}
\end{figure}

\clearpage
\section{Supplementary material: modified universal singular value thresholding~(MUSVT)}\label{sec:MUSVT}
In this section, we describe another estimator for link probability matrix $P$ based on repeated observations of a dynamic network over time via singular value thresholding. This is a modification of the universal singular value thresholding~(USVT) procedure proposed by \cite{chatterjee2015}. More specifically, as in the main text, assume $A^{(t)}$ ($t=1,\ldots,T$) such that $A_{ij}^{(t)}\sim \text{Bernoulli}(P_{ij})$ for $i\leq j$, independently. Let $\bar{A}=\sum_{t=1}^{T}A^{(t)}/T.$ 
In applying MUSVT to estimating  $P$, major steps can be summarized as follows:

\begin{itemize}
	\item[1.]Let $\bar A=\sum_{i=1}^n s_iu_iu_i^T$ be the singular value decomposition of the average adjacent matrix $\bar A$. 
	\item[2.]  Let $S=\{i: s_i\geq (2+\eta)\frac{\sqrt{n}}{\sqrt T}\}$,  where $\eta\in (0,1)$ is some small positive number. Let $\widehat{\bar A}=\sum_{i\in S} s_iu_iu_i^T$.
	
	\item [3.] Let $\widehat{P}=(\widehat{P}_{ij})$, where
	\vspace{-3mm}
	\begin{align*}
	\widehat{P}_{ij}:=\begin{cases}
	\widehat{\bar A}_{ij},&\quad \text{if}\; 0\leq \widehat{\bar A}_{ij}\leq 1\\
	1, &\quad \text{if}\; \widehat{\bar A}_{ij} \geq 1\\
	0,&\quad \text{if}\; \widehat{\bar A}_{ij} \leq 0.
	\end{cases}
	\end{align*}
\end{itemize}

$\widehat{P}$ serves as the final estimate for  $P$.  The key distinction  between our estimate and the one in \cite{chatterjee2015} is that we utilize $\bar A$, which allows us to lower the threshold level from an order of $\sqrt{n}$ to $\sqrt{n/T}$. Theorem \ref{th-usvt} quantifies the rate on $\widehat{P}$ in approximating $P$.
\begin{theorem}
	\label{th-usvt}
	Assume $P$ arises from a graphon  $f$ that is piecewise Lipschitz as defined in Definition \ref{piecewise_lipschitz}, then the following holds:
	\begin{equation}
	\label{eq-musvt}
	P\left(d_{F}(\widehat{P},P)^2\geq C(f, n, \delta) \frac{1}{\left(n^{1/3} T^{1/4}\right) }\right)\leq \epsilon(n, T), \text{ and }\epsilon(n, T)\rightarrow 0 \;\text{as}\; n, T\rightarrow \infty,
	\end{equation}
	under the condition that $T\leq n^{1-a}$ for some constant $a>0$, where $d_{F}(\cdot, \cdot)$ stands for the normalized (by $1/n$) Frobenius distance and $C(f,n, \delta)$ is a constant depending on $f$,  $n$ and $\delta$. 
\end{theorem}
\begin{proof}
The key gradients in the proof are to bound the spectral norm between $\bar A$ and $P$ as well as  the nuclear norm of $P$.
 Specifically, By Lemma 3.5 in \cite{chatterjee2015}, one has 
\begin{align}
\| \widehat{P}-P\|_{F}\leq K(\delta)\left( \|\bar A-P\| \|P\|_*\right)^{1/2},
\end{align}
By Theorem 3.4 in \cite{chatterjee2015},  one has under the conditions that $T<n^{1-a}$ for some $a>0$, 
\begin{align}
P\left(\| \bar A-P\|\geq (2+\eta)\sqrt{n/T}\right)\leq C_1(\epsilon)\exp\left ( -C_2n/T\right).
\end{align}

where $K(\delta)=(4+2\delta)\sqrt{2/\delta}+\sqrt{2+\delta},$  $\|\cdot\|_F$ is the Frobineus norm, $\|\cdot\|$ stands for the spectral norm  and $\|\cdot\|_*$ is nuclear norm.  The bound on $\|P\|_*$ is exactly the same as that in Theorem 2.7 in \cite{chatterjee2015}. Only the term  $\|\bar A-P\|$ affects the improved rate by a factor of $1/T^{1/4}$, resulting in the new rate of the theorem. 
\end{proof}

The MUSVT procedure proposed above can be used as the initial graphon estimate in designing our change-point detection algorithm.  We can prove consistency of the change-point estimation and obtain similar results as in Theorem \ref{cp_detection} for MNBS, only requiring a slightly higher threshold level and stronger conditions on the minimal true signal strength. The difference in the rates of the two quantities are affected by the quality of the initial graphon estimates. Although the rates in Theorem \ref{th-usvt} and the resulting requirements for consistency in change-point detection  are not as good as those of MNBS, MUSVT enjoys some computational advantages when the number of nodes $n$ is large, thus may still serve as an alternative practically. 

\section{Supplementary material: proof of theorems}
\begin{proposition}[Bernstein inequality]
\label{Bernstein}
Let $X_1, X_2,\ldots, X_n$ be independent zero-mean random variables. Suppose that $|X_i|\leq M$ a.s. for all $i$. Then for all positive $t$, we have
$$P\left(\sum_{i=1}^{n}X_i>t\right)\leq \exp\left(-\frac{\frac{1}{2}t^2}{\sum_{i=1}^{n}E(X_i^2)+\frac{1}{3}Mt}\right).$$
\end{proposition}

By Proposition \ref{Bernstein}, for a sequence of independent Bernoulli random variables where $X_i\sim$ Bernoulli($p_i$), we have
$$P\left(\left|\sum_{i=1}^{n}(X_i-p_i)\right|>t\right)\leq 2\exp\left(-\frac{\frac{1}{2}t^2}{\sum_{i=1}^{n}p_i(1-p_i)+\frac{1}{3}t}\right).$$

In the following, we prove the theoretical properties of MNBS by first giving two Lemmas, which extends the result of Lemmas 1 and 2 in \cite{nbhd} to the case where $T\geq 1$ repeated observations of a network is available.

Denote $I_k=[x_{k-1},x_k)$ for $1\leq k \leq K-1$ and $I_K=[x_{K-1},1]$ for the intervals in Definition \ref{piecewise_lipschitz}, and denote $\delta=\min_{1\leq k\leq K}|I_k|.$ For any $\xi\in[0,1]$, let $I(\xi)$ denote the $I_k$ that contains $\xi$. Let $S_i(\Delta)=[\xi_i-\Delta,\xi_i+\Delta]\cap I(\xi_i)$ denote the neighborhood of $\xi_i$ in which $f(x, y)$ is Lipschitz in $x\in S_i(\Delta)$ for any fixed $y$.
\begin{lemma}[Neighborhood size]
\label{nbd_size}
For any global constants $C_1>B_1>0$, define $\Delta_n=C_1\frac{\log n}{n^{1/2}\omega}$. If $\frac{n^{1/2}}{\omega}\cdot\frac{(C_1-B_1)^2}{7C_1-B_1}> \gamma+1$, there exists $\tilde{C}_1>0$ such that for $n$ large enough so that $\Delta_n<\min_k|I_k|/2,$ we have
\begin{align*}
P\left(\min_i \frac{|\{i'\neq i: \xi_{i'}\in S_i(\Delta_n)\}|}{n-1}\geq B_1\frac{\log n}{n^{1/2}\omega}\right) \geq 1-2n^{-(\tilde{C}_1+\gamma)}.
\end{align*}
\end{lemma}

\begin{proof}[Proof of Lemma \ref{nbd_size}]
For any $i$, by the definition of $S_i(\Delta_n)$, we know that $\Delta_n\leq |S_i(\Delta_n)| \leq 2\Delta_n$. By Bernstein inequality, we have
\begin{align*}
P\left(\left| \frac{|\{i'\neq i: \xi_{i'}\in S_i(\Delta_n)\}|}{n-1} -|S_i(\Delta_n)| \right|>\epsilon_n\right)&\leq
2\exp\left(-\frac{\frac{1}{2}(n-1)^2\epsilon_n^2}{(n-1)2\Delta_n+\frac{1}{3}(n-1)\epsilon_n}\right)\\
&\leq 2\exp\left(-\frac{\frac{1}{3}n\epsilon_n^2}{2\Delta_n+\frac{1}{3}\epsilon_n}\right).
\end{align*}
Take a union bound over all $i$'s gives
\begin{align*}
P\left(\max_i\left| \frac{|\{i'\neq i: \xi_{i'}\in S_i(\Delta_n)\}|}{n-1} -|S_i(\Delta_n)| \right|>\epsilon_n\right)&\leq
2n\exp\left(-\frac{\frac{1}{3}n\epsilon_n^2}{2\Delta_n+\frac{1}{3}\epsilon_n}\right).
\end{align*}
Let $\Delta_n=C_1\frac{\log n}{n^{1/2}\omega}$ and $\epsilon_n=C_2\frac{\log n}{n^{1/2}\omega}$ with $C_2=C_1-B_1>0,$ we have
\begin{align*}
&P\left(\max_i\left| \frac{|\{i'\neq i: \xi_{i'}\in S_i(\Delta_n)\}|}{n-1} -|S_i(\Delta_n)| \right|>\epsilon_n\right)\\
\leq&
2n\exp\left(-\frac{\frac{1}{3}n\epsilon_n^2}{2\Delta_n+\frac{1}{3}\epsilon_n}\right)\leq 2n\exp\left(-\frac{\frac{1}{3}nC_2^2\frac{\log n}{n^{1/2}\omega}}{2C_1+\frac{1}{3}C_2}\right)\\
=&2n^{1-\frac{n^{1/2}}{\omega}\frac{C_2^2}{6C_1+C_2}}\leq 2n^{-(\tilde{C}_1+\gamma)},
\end{align*}
for some $\tilde{C}_1>0$ as long as $\frac{n^{1/2}}{\omega}\cdot\frac{C_2^2}{6C_1+C_2}=\frac{n^{1/2}}{\omega}\cdot\frac{(C_1-B_1)^2}{7C_1-B_1}> 1+\gamma$. Thus, with probability $1-2n^{-(\tilde{C}_1+\gamma)}$, we have
\begin{align*}
\min_i \frac{|\{i'\neq i: \xi_{i'}\in S_i(\Delta_n)\}|}{n-1}& \geq \min_i S_i(\Delta_n)-\epsilon_n \geq \Delta_n-\epsilon_n \\
&= (C_1-C_2)\frac{\log n}{n^{1/2}\omega}=B_1\frac{\log n}{n^{1/2}\omega}.
\end{align*}
This completes the proof of Lemma \ref{nbd_size}.
\end{proof}

\begin{lemma}[Neighborhood approximation]
	\label{nbd_property}
	Suppose we select the neighborhood $\mathcal{N}_i$ by thresholding at the lower $q$-th quantile of $\{\tilde{d}(i,k):k\neq i\}$, where we set $q=B_0\frac{\log n}{n^{1/2}\omega}$ with $0<B_0\leq B_1$ for the $B_1$ from Lemma \ref{nbd_size}. For any global constant $C_3>0,$ if $\frac{C_3^2\log n}{6}\cdot\frac{T}{\omega^2}>2+\gamma$ and $\frac{C_3}{2}\cdot\frac{n^{1/2}}{\omega} >2+\gamma$ hold, there exists $\tilde{C}_2>0$ such that if $n$ is large enough so that (i) all conditions on $n$ in Lemma \ref{nbd_size} hold; (ii) $B_1\frac{n^{1/2}\log n}{\omega}\geq 4$, then the neighborhood $\mathcal{N}_i$ has the following properties:
	\begin{enumerate}
		\item $|\mathcal{N}_i|\geq B_0\frac{n^{1/2}\log n}{\omega}$.
		\item With probability $1-2n^{-(\tilde{C}_1+\gamma)}-2n^{-(\tilde{C}_2+\gamma)}$, for all $i$ and $i'\in\mathcal{N}_i$, we have
		$$\|P_{i\cdot}-P_{i'\cdot}\|_2^2/n \leq (6LC_1+24C_3)\frac{\log n}{n^{1/2}\omega}.$$
	\end{enumerate}
\end{lemma}

\begin{proof}[Proof of Lemma \ref{nbd_property}]
	The first claim follows immediately from the definition of quantile and $q$, since $|\mathcal{N}_i|\geq n\cdot q=nB_0\frac{\log n}{n^{1/2}\omega}=B_0\frac{n^{1/2}\log n}{\omega}$.
	
	To prove the second claim, we first give a concentration result. For any $i,j$ such that $i\neq j$, we have
	\begin{align*}
	&\left|(\bar{A}^2/n)_{ij}-(P^2/n)_{ij}\right|=\left|\sum_k (\bar{A}_{ik}\bar{A}_{kj}-P_{ik}P_{kj})\right|/n\\
	\leq & \frac{\left|\sum_{k\neq i,j} (\bar{A}_{ik}\bar{A}_{kj}-P_{ik}P_{kj}) \right|}{n-2}\cdot\frac{n-2}{n} + 
	 \frac{\left|(\bar{A}_{ii}\bar{A}_{ij}-P_{ii}P_{ij}) \right|}{n} +  \frac{\left|(\bar{A}_{ij}\bar{A}_{jj}-P_{ij}P_{jj}) \right|}{n}.
	\end{align*}
	We can easily show that
	$$\text{Var}(\bar{A}_{ik}\bar{A}_{kj})=\frac{P_{ik}^2P_{kj}(1-P_{kj}) + P_{kj}^2P_{ik}(1-P_{ik})}{T}+\frac{P_{ik}(1-P_{ik})P_{kj}(1-P_{jk})}{T^2}\leq \frac{1}{T}.$$
	Thus, by the independence among $\bar{A}_{ik}\bar{A}_{kj}$ and Bernstein inequality, we have
	\begin{align*}
	P\left(\frac{\left|\sum_{k\neq i,j} (\bar{A}_{ik}\bar{A}_{kj}-P_{ik}P_{kj}) \right|}{n-2} \geq \epsilon_n \right)\leq 2\exp\left(-\frac{\frac{1}{2}(n-2)^2\epsilon_n^2}{(n-2)\frac{1}{T}+\frac{1}{3}(n-2)\epsilon_n}\right)\leq 2\exp\left(-\frac{\frac{1}{3}n\epsilon_n^2}{\frac{1}{T}+\frac{1}{3}\epsilon_n}\right).
	\end{align*}
	Take a union bound over all $i\neq j$, we have
	\begin{align}
	\label{A2_concentration}
	&P\left(\max_{i,j:i\neq j}\frac{\left|\sum_{k\neq i,j} (\bar{A}_{ik}\bar{A}_{kj}-P_{ik}P_{kj}) \right|}{n-2} \geq \epsilon_n \right)\leq 2n^2\exp\left(-\frac{\frac{1}{3}n\epsilon_n^2}{\frac{1}{T}+\frac{1}{3}\epsilon_n}\right)
	\\\leq& 2n^2\max\left(\exp\left(-\frac{1}{6}nT\epsilon_n^2 \right),\exp\left(-\frac{1}{2}n\epsilon_n \right)\right).\nonumber
	\end{align}
	Let $\epsilon_n=C_3\frac{\log n}{n^{1/2}\omega}$, we have
    \begin{align*}
    &2n^2\exp\left(-\frac{1}{6}nT\epsilon_n^2 \right)=2n^2\exp\left(-\frac{1}{6}nTC_3^2\frac{(\log n)^2}{n\omega^2}\right)=2n^{2-\frac{C_3^2\log n}{6}\cdot\frac{T}{\omega^2}}\leq 2n^{-(\tilde{C}_2+\gamma)}/3,\\
    &2n^2\exp\left(-\frac{1}{2}n\epsilon_n \right) = 2n^2\exp\left(-\frac{1}{2}n C_3\frac{\log n}{n^{1/2}\omega} \right)=2n^{2-\frac{C_3n^{1/2}}{2\omega}}\leq 2n^{-(\tilde{C}_2+\gamma)}/3,
    \end{align*}
    for some $\tilde{C}_2>0$ as long as $\frac{C_3^2\log n}{6}\cdot\frac{T}{\omega^2}>2+\gamma$ and $\frac{C_3}{2}\cdot\frac{n^{1/2}}{\omega} >2+\gamma.$
    
    Similarly, we have
    \begin{align*}
    P\left(\max_{i,j:i\neq j}  \frac{\left|(\bar{A}_{ii}\bar{A}_{ij}-P_{ii}P_{ij}) \right|}{n} >\epsilon_n \right) \leq 2n^2\exp\left(-\frac{\frac{1}{3}n^2\epsilon_n^2}{\frac{1}{T}+\frac{1}{3}n\epsilon_n}\right)\leq 2n^{-(\tilde{C}_2+\gamma)}/3,
    \end{align*}
    as long as $\frac{C_3^2\log n}{6}\cdot\frac{n}{\omega^2}\cdot T>2+\gamma$ and $\frac{C_3}{2}\cdot\frac{n^{1/2}}{\omega} >2+\gamma$
    
    Thus, combine the above results, we have that with probability $1-2n^{-(\tilde{C}_2+\gamma)}$,
    $$\max_{i,j:i\neq j}\left|(\bar{A}^2/n)_{ij}-(P^2/n)_{ij}\right|\leq 3\epsilon_n=3C_3\frac{\log n}{n^{1/2}\omega}.$$
    
    Following the same argument as \cite{nbhd}, we have that for all $i$ and any $\tilde{i}$ such that $\xi_{\tilde{i}}\in S_i(\Delta_n)$,
    \begin{align*}
    \left|(P^2/n)_{ik}-(P^2/n)_{\tilde{i}k} \right| = \left|\langle P_{i\cdot},P_{k\cdot}\rangle-\langle P_{\tilde{i}\cdot},P_{k\cdot}\rangle\right|/n
    \leq \|P_{i\cdot}-P_{\tilde{i}\cdot}\|_2\|P_{k\cdot}\|_2/n\leq L\Delta_n,
    \end{align*}
    for all $k=1,\ldots,n$, where the last inequality follows from the piecewise Lipschitz condition of the graphon such that
    $$\left|P_{\tilde{i}l}-P_{il}\right|=\left|f(\xi_{\tilde{i}},\xi_l)-f(\xi_i,\xi_l)\right|\leq L|\xi_{\tilde{i}}-\xi_i|\leq L\Delta_n \text{ for all } l=1,\ldots,n, $$
    and from $\|P_{k\cdot}\|_2\leq n^{1/2}$ for all $k.$
    
    We now try to upper bound $\tilde{d}(i,i')$ for all $i'\in \mathcal{N}_i$. We first bound $\tilde{d}(i,\tilde{i})$ for all $\tilde{i}$ with $\xi_{\tilde{i}}\in S_i(\Delta_n)$ simultaneously. By above, we know that with probability $1-2n^{-(\tilde{C}_2+\gamma)}$, we have
    \begin{align*}
    \tilde{d}(i,\tilde{i})&=\max_{k\neq i,\tilde{i}}\left|(\bar{A}^2/n)_{ik}-(\bar{A}^2/n)_{\tilde{i}k}\right|\leq \max_{k\neq i,\tilde{i}}\left|(P^2/n)_{ik}-(P^2/n)_{\tilde{i}k}\right| + 2\max_{i,j:i\neq j} \left|(\bar{A}^2/n)_{ij}-(P^2/n)_{ij}\right|\\
    & \leq L\Delta_n + 6C_3\frac{\log n}{n^{1/2}\omega},
    \end{align*}
    for all $i$ and any $\tilde{i}$ such that $\xi_{\tilde{i}}\in S_i(\Delta_n)$.
    
    By the above result and Lemma \ref{nbd_size}, we know that with probability $1-2n^{-(\tilde{C}_1+\gamma)}-2n^{-(\tilde{C}_2+\gamma)}$, for all $i,$ at least $B_1\frac{\log n}{n^{1/2}\omega}$ fraction of nodes $\tilde{i}\neq i$ satisfy both $\xi_{\tilde{i}}\in S_i(\Delta_n)$ and $\tilde{d}(i,\tilde{i})\leq L\Delta_n + 6C_3\frac{\log n}{n^{1/2}\omega}$. Thus we have
    $$\tilde{d}(i,i')\leq L\Delta_n + 6C_3\frac{\log n}{n^{1/2}\omega}$$
    holds for all $i$ and all $i'\in \mathcal{N}_i$ simultaneously with probability $1-2n^{-(\tilde{C}_1+\gamma)}-2n^{-(\tilde{C}_2+\gamma)}$, since by definition nodes in $\mathcal{N}_i$ have the lowest $q=B_0\frac{\log n}{n^{1/2}\omega}\leq B_1\frac{\log n}{n^{1/2}\omega}$ fraction of values in $\{\tilde{d}(i,k), k\neq i\}$.
    
    We are now ready to complete the proof of the second claim of Lemma \ref{nbd_property}. With probability $1-2n^{-(\tilde{C}_1+\gamma)}-2n^{-(\tilde{C}_2+\gamma)}$, for $n$ large enough such that $\min_i |\{i'\neq i: \xi_{i'}\in S_i(\Delta_n)\}|\geq B_1\frac{n^{1/2}\log n}{\omega}\geq 4$~(by Lemma \ref{nbd_size}), we have that for all $i$ and $i'\in \mathcal{N}_i$, we can find $\tilde{i}\in S_i(\Delta_n)$, $\tilde{i'}\in S_{i'}(\Delta_n)$ such that $\tilde{i}, i, \tilde{i'}, i'$ are different from each other and
    \begin{align*}
    \|P_{i\cdot}-P_{i'\cdot}\|_2^2/n &= (P^2/n)_{ii}-(P^2/n)_{i'i} + (P^2/n)_{i'i'} -(P^2/n)_{ii'}\\
    &\leq |(P^2/n)_{ii}-(P^2/n)_{i'i}| + |(P^2/n)_{i'i'} -(P^2/n)_{ii'}|\\
    &\leq |(P^2/n)_{i\tilde{i}}-(P^2/n)_{i'\tilde{i}}| + |(P^2/n)_{i'\tilde{i'}} -(P^2/n)_{i\tilde{i'}}|+4L\Delta_n\\
    &\leq |(\bar{A}^2/n)_{i\tilde{i}}-(\bar{A}^2/n)_{i'\tilde{i}}| + |(\bar{A}^2/n)_{i'\tilde{i'}} -(\bar{A}^2/n)_{i\tilde{i'}}|+4L\Delta_n + 12C_3\frac{\log n}{n^{1/2}\omega}\\
    &\leq 2\max_{k\neq i,i'} |(\bar{A}^2/n)_{ik}-(\bar{A}^2/n)_{i'k}| +4L\Delta_n + 12C_3\frac{\log n}{n^{1/2}\omega}\\
    &=2\tilde{d}(i,i')  +4L\Delta_n + 12C_3\frac{\log n}{n^{1/2}\omega}\leq 6L\Delta_n + 24 C_3\frac{\log n}{n^{1/2}\omega}\\
    &=(6LC_1+24C_3)\frac{\log n}{n^{1/2}\omega}.
    \end{align*}
    This completes the proof of Lemma \ref{nbd_property}.
\end{proof}

Based on Lemma \ref{nbd_size} and \ref{nbd_property}, we are now ready to prove Theorem \ref{MNBS_convergence}, which provides the error bound for MNBS.

\begin{proof}[Proof of Theorem \ref{MNBS_convergence}]
	To prove Theorem \ref{MNBS_convergence}, it suffices to show that with high probability, the following holds for all $i$.
	$$\frac{1}{n}\sum_j (\tilde{P}_{ij}-P_{ij})^2\leq C\cdot \frac{\log n}{n^{1/2}\omega}.$$
	
	We first perform a bias-variance decomposition such that
	\begin{align*}
	&\frac{1}{n}\sum_j (\tilde{P}_{ij}-P_{ij})^2=\frac{1}{n}\sum_j \left\{\frac{\sum_{i'\in\mathcal{N}_i}(\bar{A}_{i'j}-P_{ij})}{|\mathcal{N}_i|} \right\}^2\\
	=&\frac{1}{n}\sum_j \left\{\frac{\sum_{i'\in\mathcal{N}_i}(\bar{A}_{i'j}-P_{i'j})+(P_{i'j}-P_{ij})}{|\mathcal{N}_i|} \right\}^2\\
	\leq &\frac{2}{n}\sum_j \left\{\frac{\sum_{i'\in\mathcal{N}_i}(\bar{A}_{i'j}-P_{i'j})}{|\mathcal{N}_i|} \right\}^2 + \frac{2}{n}\sum_j \left\{\frac{\sum_{i'\in\mathcal{N}_i}(P_{i'j}-P_{ij})}{|\mathcal{N}_i|} \right\}^2\\
	=&2\cdot\frac{1}{n}\sum_j\left\{J_1(i,j) + J_2(i,j)\right\}.
	\end{align*}
	
	Thus, our goal is to bound $\frac{1}{n}\sum_j J_1(i,j)$ and $\frac{1}{n}\sum_jJ_2(i,j)$. We first bound $n^{-1}\sum_j J_1(i,j)$. We have
	\begin{align*}
	\frac{1}{n}\sum_j J_1(i,j)=&\frac{1}{n|\mathcal{N}_i|^2}\sum_j \left\{\sum_{i'\in\mathcal{N}_i}(\bar{A}_{i'j}-P_{i'j}) \right\}^2\\
	=&\frac{1}{n|\mathcal{N}_i|^2}\sum_j \left\{\sum_{i'\in\mathcal{N}_i}(\bar{A}_{i'j}-P_{i'j})^2 + \sum_{i'\in\mathcal{N}_i}\sum_{i''\neq i',i''\in\mathcal{N}_i}(\bar{A}_{i'j}-P_{i'j})(\bar{A}_{i''j}-P_{i''j}) \right\}\\
	=&\frac{1}{|\mathcal{N}_i|^2}\sum_{i'\in\mathcal{N}_i}\left\{ \frac{1}{n}\sum_j(\bar{A}_{i'j}-P_{i'j})^2+ \frac{1}{n}\sum_j\sum_{i''\neq i',i''\in\mathcal{N}_i}(\bar{A}_{i'j}-P_{i'j})(\bar{A}_{i''j}-P_{i''j}) \right\}.
	\end{align*}
	
	For the first term, by Bernstein inequality and the fact that $\text{Var}\left[(\bar{A}_{i'j}-P_{i'j})^2\right]\leq 1/T^2$ for all $i',j$, we have
	\begin{align*}
	&P\left(\frac{1}{n}\left|\sum_j \left[(\bar{A}_{i'j}-P_{i'j})^2-E(\bar{A}_{i'j}-P_{i'j})^2 \right] \right|>\epsilon_n \right)\\
	=&P\left(\frac{1}{n}\left|\sum_j \left[(\bar{A}_{i'j}-P_{i'j})^2-\frac{P_{i'j}(1-P_{i'j})}{T} \right] \right|>\epsilon_n \right)\\
	\leq & 2 \exp\left(-\frac{\frac{1}{2}n^2\epsilon_n^2}{\sum_{j=1}^{n}\frac{1}{T^2}+\frac{1}{3}n\epsilon_n}\right)=
	2 \exp\left(-\frac{\frac{1}{2}n\epsilon_n^2}{\frac{1}{T^2}+\frac{1}{3}\epsilon_n}\right).
	\end{align*}
	Let $\epsilon_n=C_4\left(\frac{\log n}{\omega}\right)^2$, 
	by union bound, there exists $\tilde{C}_3>0$ such that
	\begin{align*}
		&P\left(\max_{i'}\frac{1}{n}\left|\sum_j \left[(\bar{A}_{i'j}-P_{i'j})^2-E(\bar{A}_{i'j}-P_{i'j})^2 \right] \right|>\epsilon_n \right)\\
		\leq& 2n\exp\left(-\frac{\frac{1}{2}n\epsilon_n^2}{\frac{1}{T^2}+\frac{1}{3}\epsilon_n}\right) \leq 2n \max\left( \exp\left(-\frac{1}{4}n\epsilon_n^2T^2\right),
		\exp\left(-\frac{3}{4}n\epsilon_n\right) \right)\leq 2n^{-(\tilde{C}_3+\gamma)},
	\end{align*}
	for any $\gamma>0$ as long as $\frac{C_4^2 (\log n)^3}{4}\cdot \frac{n}{\omega^2}\cdot \frac{T^2}{\omega^2}>(1+\gamma)$ and $\frac{3C_4\log n}{4}\cdot \frac{n}{\omega^2}>1+\gamma$. Thus, with probability $1-2n^{-(\tilde{C}_3+\gamma)}$, we have
	$$\max_{i'}\frac{1}{n}\sum_j(\bar{A}_{i'j}-P_{i'j})^2\leq \frac{1}{T}+C_4\left(\frac{\log n}{\omega}\right)^2.$$
	
	For the second term, we have
	\begin{align*}
	&\frac{1}{n}\sum_j\sum_{i'\in\mathcal{N}_i}\sum_{i''\neq i',i''\in\mathcal{N}_i}(\bar{A}_{i'j}-P_{i'j})(\bar{A}_{i''j}-P_{i''j}) \\
	\leq &\sum_{i'\in\mathcal{N}_i}\sum_{i''\neq i',i''\in\mathcal{N}_i}\left|\frac{1}{n}\sum_j(\bar{A}_{i'j}-P_{i'j})(\bar{A}_{i''j}-P_{i''j})\right| \\
	\leq &\sum_{i'\in\mathcal{N}_i}\sum_{i''\neq i',i''\in\mathcal{N}_i}\left\{\left|\frac{1}{n-2}\sum_{j\neq i',i''}(\bar{A}_{i'j}-P_{i'j})(\bar{A}_{i''j}-P_{i''j})\right|\cdot\frac{n-2}{n} \right.\\
	&\left.+\frac{\left|(\bar{A}_{i'i'}-P_{i'i'})(\bar{A}_{i''i'}-P_{i''i'})\right|}{n}+\frac{\left|(\bar{A}_{i'i''}-P_{i'i''})(\bar{A}_{i''i''}-P_{i''i''})\right|}{n}
	\right\}.
	\end{align*}
	Note that for any $i'\neq i''$, $\text{Var}\left[ (\bar{A}_{i'j}-P_{i'j})(\bar{A}_{i''j}-P_{i''j})\right]\leq 1/T^2$. Similar to the proof of Lemma \ref{nbd_property}, via Bernstein inequality and union bound, we can show that there exists $\tilde{C}_4>0$ such that
	$$P\left(\max_{i',i'':i'\neq i''}\left|\frac{1}{n}\sum_j(\bar{A}_{i'j}-P_{i'j})(\bar{A}_{i''j}-P_{i''j})\right|\leq 3C_5\frac{\log n}{n^{1/2}\omega}\right)\geq1-2n^{-(\tilde{C}_4+\gamma)} ,$$
	as long as $\frac{C_5^2\log n}{6}\cdot\frac{T^2}{\omega^2}>2+\gamma$ and $\frac{C_5}{2}\cdot\frac{n^{1/2}}{\omega} >2+\gamma.$
	
	Combine above results, we have that with probability $1-2n^{-(\tilde{C}_3+\gamma)}-2n^{-(\tilde{C}_4+\gamma)}$,
	\begin{align*}
	\frac{1}{n}\sum_jJ_1(i,j)&\leq \frac{1}{|\mathcal{N}_i|^2}\sum_{i'\in\mathcal{N}_i}\left(\frac{1}{T}+C_4 \frac{\log n}{\omega^2} + (|\mathcal{N}_i|-1) 3C_5\frac{\log n}{n^{1/2}\omega}\right)\\
	&\leq \frac{1}{|\mathcal{N}_i|}\left(\frac{1}{T}+C_4 \left(\frac{\log n}{\omega}\right)^2 \right) + 3C_5 \frac{\log n}{n^{1/2}\omega}\\
	&\leq \left(\frac{1}{B_0}\cdot\frac{\omega^2}{T(\log n)^2} +C_4/B_0+3C_5 \right) \frac{\log n}{n^{1/2}\omega}\\
	&\leq \left((1+C_4)/B_0+3C_5 \right) \frac{\log n}{n^{1/2}\omega},
	\end{align*}
	as long as $\frac{\omega^2}{T(\log n)^2}\leq 1.$
	
	We now bound $n^{-1}\sum_jJ_2(i,j)$. By Lemma \ref{nbd_property}, with probability $1-2n^{-(\tilde{C}_1+\gamma)}-2n^{-(\tilde{C}_2+\gamma)}$, for all $i$ simultaneously, we have
	\begin{align*}
	\frac{1}{n}\sum_j J_2(i,j)&=\frac{1}{n}\sum_j \left\{\frac{\sum_{i'\in\mathcal{N}_i}(P_{i'j}-P_{ij})}{|\mathcal{N}_i|} \right\}^2 \leq \frac{1}{n}\frac{\sum_j\sum_{i'\in\mathcal{N}_i} (P_{i'j}-P_{ij})^2}{|\mathcal{N}_i|}\\
	&=\frac{\sum_{i'\in\mathcal{N}_i}\|P_{i'\cdot}-P_{i\cdot}\|_2^2/n }{|\mathcal{N}_i|}\leq (6LC_1+24C_3)\frac{\log n}{n^{1/2}\omega},
	\end{align*}
	where the first inequality follows from Cauchy-Schwarz inequality.
	
	Thus, together, with probability $1-2n^{-(\tilde{C}_1+\gamma)}-2n^{-(\tilde{C}_2+\gamma)}-2n^{-(\tilde{C}_3+\gamma)}-2n^{-(\tilde{C}_4+\gamma)}$, we have for all $i$ simultaneously
	\begin{align}
	\label{MNBS_finalproof}
	\frac{1}{n}\sum_j (\tilde{P}_{ij}-P_{ij})^2\leq 2\left((1+C_4)/B_0+3C_5+ 6LC_1+24C_3\right)\frac{\log n}{n^{1/2}\omega}=C\frac{\log n}{n^{1/2}\omega}.
	\end{align}
	
	Collecting all the conditions on $n, T, \omega$ for \eqref{MNBS_finalproof} to hold, we have 
	\begin{enumerate}
		\item Lemma \ref{nbd_size}: $\frac{n^{1/2}}{\omega}\cdot\frac{(C_1-B_1)^2}{7C_1-B_1}> \gamma+1$ and $C_1>B_1>0$;
		\item Lemma \ref{nbd_property}: $\frac{C_3^2\log n}{6}\cdot\frac{T}{\omega^2}>2+\gamma$ and $\frac{C_3}{2}\cdot\frac{n^{1/2}}{\omega} >2+\gamma$ and $B_1\geq B_0$;
		\item $\frac{C_4^2 (\log n)^3}{4}\cdot \frac{n}{\omega^2}\cdot \frac{T^2}{\omega^2}>(1+\gamma)$ and $\frac{3C_4\log n}{4}\cdot \frac{n}{\omega^2}>1+\gamma$; \\
		      $\frac{C_5^2\log n}{6}\cdot\frac{T^2}{\omega^2}>2+\gamma$ and $\frac{C_5}{2}\cdot\frac{n^{1/2}}{\omega} >2+\gamma$;\\
		      $\frac{\omega^2}{T(\log n)^2}\leq 1.$
	\end{enumerate}
	
	It is easy to see that, for any $\gamma>0$ and $B_0>0$, we can always find $B_1,C_1,C_3,C_4,C_5$ such that all inequalities in (1)-(3) hold for all $n$ large enough as long as $\omega \leq \min(n^{1/2},(T\log n)^{1/2})$. Take $\omega = \min(n^{1/2},(T\log n)^{1/2})$, this completes the proof of Theorem \ref{MNBS_convergence}.
\end{proof}

\begin{proof}[Proof of Theorem \ref{cp_detection}]
	Denote $\bar{P}_{t1,h}=\sum_{i=t-h+1}^{t}P^{(i)}/h$ and $\bar{P}_{t2,h}=\sum_{i=t+1}^{t+h}P^{(i)}/h$. For each $t=1,\ldots, T$, $\tilde{P}_{t1,h}$ and $\tilde{P}_{t2,h}$ are MNBS estimators for $\bar{P}_{t1,h}$ and $\bar{P}_{t2,h}$ respectively. 
	
	We call $t$ an $h$-flat point if there is no change-point within $\{t-h+1,\ldots,t+h-1\}$. The main idea of the proof is to analyze the behavior of $\tilde{P}_{t1,h}$ and $\tilde{P}_{t2,h}$ for all $h$-flat points and all true change-points among $t=1,\ldots,T$.
	
	The key observation is that for both an $h$-flat point and an true change-point, the adjacency matrices $\{A^{(i)}\}_{i=t-h+1}^t$ or $\{A^{(i)}\}_{i=t+1}^{t+h}$ that are used in the estimation of $\tilde{P}_{t1,h}$ or $\tilde{P}_{t2,h}$ are generated by the same probability matrix $P$ and thus the result of Theorem \ref{MNBS_convergence} can be directly applied.
	
	By assumption we have $(h\log n)^{1/2}< n^{1/2}$, thus $\omega=\min(n^{1/2}, (h\log n)^{1/2})=(h\log n)^{1/2}$. Thus by Theorem \ref{MNBS_convergence}, for any $t$ that is an $h$-flat point, we have
	\begin{align*}
	&P(D(t,h)>\Delta_D)=P(d_{2,\infty}(\tilde{P}_{t1,h},\tilde{P}_{t2,h})^2>\Delta_D)\\
	\leq & P(d_{2,\infty}(\tilde{P}_{t1,h},\bar{P}_{t1,h})^2+ d_{2,\infty}(\tilde{P}_{t2,h},\bar{P}_{t2,h})^2>\Delta_D/2)\\
	\leq & P(d_{2,\infty}(\tilde{P}_{t1,h},\bar{P}_{t1,h})^2>\Delta_D/4) + P(d_{2,\infty}(\tilde{P}_{t2,h},\bar{P}_{t2,h})^2>\Delta_D/4)\leq 2n^{-\gamma},
	\end{align*}
	where the second to last inequality uses the fact that $\bar{P}_{t1,h}=\bar{P}_{t2,h}$ for an $h$-flat point, and the last inequality follows from Theorem \ref{MNBS_convergence} and the fact that $\Delta_D/(C(\log n)^{1/2}/(n^{1/2}h^{1/2}))\to \infty$ for any $C>0.$
	
	For any $t$ that is a true change-point, we have
	\begin{align*}
	&P(D(t,h)>\Delta_D)=P(d_{2,\infty}(\tilde{P}_{t1,h},\tilde{P}_{t2,h})^2>\Delta_D)\\
	\geq & P(d_{2,\infty}(\bar{P}_{t1,h},\bar{P}_{t2,h})- d_{2,\infty}(\tilde{P}_{t1,h},\bar{P}_{t1,h})-d_{2,\infty}(\tilde{P}_{t2,h},\bar{P}_{t2,h})>\sqrt{\Delta_D})\\
	= & P(d_{2,\infty}(\tilde{P}_{t1,h},\bar{P}_{t1,h})+d_{2,\infty}(\tilde{P}_{t2,h},\bar{P}_{t2,h})< d_{2,\infty}(\bar{P}_{t1,h},\bar{P}_{t2,h})-\sqrt{\Delta_D})\\
	\geq & P(d_{2,\infty}(\tilde{P}_{t1,h},\bar{P}_{t1,h})+d_{2,\infty}(\tilde{P}_{t2,h},\bar{P}_{t2,h})< \sqrt{\Delta_D}(\sqrt{\Delta^*/\Delta_D}-1))\\
	\geq & 1- P(d_{2,\infty}(\tilde{P}_{t1,h},\bar{P}_{t1,h})^2> \Delta_D(\sqrt{\Delta^*/\Delta_D}-1)^2/4) \\
	&~~~ - P(d_{2,\infty}(\tilde{P}_{t2,h},\bar{P}_{t2,h})^2> \Delta_D(\sqrt{\Delta^*/\Delta_D}-1)^2/4) \geq 1- 2n^{-\gamma},
	\end{align*}
	where the second inequality uses the fact that $d_{2,\infty}(\bar{P}_{t1,h},\bar{P}_{t2,h})\geq \sqrt{\Delta^*}$ for a true change-point, and the last inequality follows from Theorem \ref{MNBS_convergence} and the fact that $\Delta_D/(C(\log n)^{1/2}/(n^{1/2}h^{1/2}))\to \infty$ for any $C>0.$
	
	Let $\mathcal{F}_h$ be the set of all flat points $t$ and $\mathcal{J}$ be the set of all true change-points. Consider the event $\mathcal{A}_\tau=\{D(\tau,h)>\Delta_D \}$ for true change-points $\tau\in\mathcal{J}$ and the event $\mathcal{B}_t=\{D(t,h)<\Delta_D\}$ for flat points $t\in \mathcal{F}_h$. Define the event
	\begin{align*}
	\mathcal{\xi}_n=\left(\bigcap_{\tau \in\mathcal{J}}\mathcal{A}_\tau \right)\bigcap \left(\bigcap_{t \in\mathcal{F}_h}\mathcal{B}_t \right).
	\end{align*}
	By the above result, we have that
	\begin{align*}
	P(\mathcal{\xi}_n)=1-P(\mathcal{\xi}_n^c)\geq 1-P\left(\bigcup_{\tau\in\mathcal{J}}A_\tau^c\right)-P\left(\bigcup_{t\in\mathcal{F}_h}\mathcal{B}_t^c\right)\geq 1-2Tn^{-\gamma}\to 1,
	\end{align*}
	as long as $Tn^{-\gamma}\to 0$.
	
	We now prove that $\xi_n$ implies the event $\{\hat{J}=J\} \cap \{\mathcal{J}\subset: \hat{\mathcal{J}}\pm h\}$. Under $\xi_n$, no flat point will be selected at the thresholding steps. Thus, for any point $\hat{\tau}\in\mathcal{J}$, there is at least one change-point in its neighborhood $\{\hat{\tau}-h+1,\ldots, \hat{\tau}+h-1\}$. On the other hand, by assumption $h<D^*/2$, thus, there exists at most one change-point in $\{\hat{\tau}-h+1,\ldots, \hat{\tau}+h-1\}$. Together, it implies that there is exactly one change-point in $\{\hat{\tau}-h+1,\ldots, \hat{\tau}+h-1\}$ for each $\hat{\tau}\in\hat{\mathcal{J}}.$
	
	Meanwhile, under $\xi_n$, for every true change-point $\tau \in \mathcal{J}$, we have $D(\tau,h)>\Delta_D$. Note that $\tau-h$ and $\tau+h$ are $h$-flat points since $h<D^*/2$, thus $\max(D(\tau+h,h),D(\tau-h,h))<\Delta_D.$ Thus, for every true change-point $\tau,$ there exists a local maximizer, say $\hat{\tau}$, which is in $\{\tau-h+1,\ldots,\tau+h-1 \}$ with $D(\hat{\tau},h)\geq D(\tau,h)>\Delta_D.$
	
	Combining the above result, we have that $$P\left(\{\hat{J}=J\} \cap \{\mathcal{J}\subset: \hat{\mathcal{J}}\pm h\}\right) \to 1.$$
\end{proof}


\end{document}